\documentclass{comjnl}
\usepackage{algorithm}
\usepackage{algpseudocode}
\usepackage{amsmath,amssymb}
\usepackage{color}
\usepackage{subcaption}
\usepackage{amssymb}
\usepackage{amsmath}\usepackage[scr=rsfs]{mathalfa}
\usepackage{enumerate}
\usepackage{array}
\usepackage{changepage}
\usepackage{algorithm}
\usepackage{algpseudocode}
\usepackage{amsmath,amssymb}
\usepackage{color}
\usepackage{subcaption}
\usepackage{amssymb}
\usepackage{amsmath}\usepackage[scr=rsfs]{mathalfa}
\usepackage{enumerate}
\usepackage{amsthm}
\usepackage{adjustbox}
\usepackage{array}
\usepackage{changepage}
\usepackage{soul}
\usepackage{stfloats}

\begin{document}
\title{Gathering Autonomous Mobile Robots Under the Adversarial Defected View Model}
\author{Prakhar Shukla}
\affiliation{Indian Institute of Technology Jodhpur, India}

\author{Seshunadh Tanuj Peddinti}
\affiliation{Sardar VallabhBhai National Institute of Technology Surat, India}

\author{Subhash Bhagat}
\affiliation{Indian Institute of Technology Jodhpur, India}

\email{p23ma0011@iitj.ac.in$^1$, sbhagat@iitj.ac.in$^1$, i22ma045@amhd.svnit.ac.in$^2$}

\shortauthors{Shukla et al.}


\begin{abstract}
This paper studies the gathering problem for a set of $N \geq 2$ autonomous mobile robots operating in the Euclidean plane under the distributed {\it Look–Compute–Move} model. We consider oblivious robots executing under the adversarial {\it defected view} model, in which an activated robot may observe only a restricted subset of robots due to adversarial visibility faults. Consequently, the information obtained during each Look phase may be incomplete and dynamically altered. The objective is to guarantee deterministic finite-time gathering at a location not known a priori despite such sensing restrictions. We present two distributed algorithms under distinct scheduling assumptions. In the fully synchronous ($\mathcal{FSYNC}$) model, we prove finite-time gathering in the adversarial $(4,2)$ defected view setting, resolving a previously open case without requiring additional capabilities or coordinate agreement. In the asynchronous ($\mathcal{ASYNC}$) model, we establish finite-time gathering under the general adversarial $(N,K)$ defected view model, where an activated robot observes at most $K$ of the other $N-1$ robots for any $1 \leq K < N-1$. Both results hold under non-rigid motion. The proposed algorithm for the $\mathcal{ASYNC}$ model assumes agreement in the direction and orientation of one coordinate axis.
\end{abstract}
\maketitle

 \section{Introduction} 
Swarm robotics studies how a large number of simple, autonomous robots can collectively perform coordinated tasks without relying on centralized control or global supervision. Inspired by natural systems such as ant colonies, bird flocks, and fish schools, swarm robotic systems emphasize decentralization, scalability, robustness, and fault tolerance. Each robot in the swarm typically has limited sensing, computation, and communication capabilities; however, through local interactions, the swarm as a whole can exhibit complex and reliable collective behavior.

Due to these characteristics, swarm robotic systems are particularly suitable for deployment in environments that are inaccessible, hazardous, or unpredictable for humans. Representative application domains include operations in radioactive or contaminated zones, agricultural monitoring, security and surveillance, search-and-rescue missions, and exploration of unknown or unstructured environments. Moreover, the use of simple and identical robots makes such systems cost-effective and easily scalable, enabling large-scale deployment when required.

In swarm robotics, robots are often assumed to be \textbf{anonymous}, meaning they lack unique identifiers and cannot be distinguished by their appearance or behavior. They are typically modeled as \textbf{oblivious}, retaining no memory of past actions or observations, and operate solely based on their current perception of the environment. Robots do not communicate explicitly by exchanging messages; instead, coordination is achieved implicitly through observation of other robots’ positions. These assumptions significantly limit individual robot capabilities, making coordination problems both challenging and theoretically interesting.

A widely studied abstraction for robot behavior is the \textbf{Look--Compute--Move} cycle. In this model, an active robot first observes the positions of other robots (\textbf{Look}), then computes a destination based on the observed configuration (\textbf{Compute}), and finally moves toward the computed destination (\textbf{Move}). Robot activations and movements are often controlled by an adversarial scheduler. In the \textbf{asynchronous} ($\mathcal{ASYNC}$) model, activations occur independently, and the duration of each phase is finite but unpredictable. This model captures realistic timing uncertainties and is considered one of the most general and challenging execution models.

Another important aspect of robot coordination is the movement assumption. Under the \textbf{rigid} movement model, a robot always reaches its intended destination in a single move. In contrast, the \textbf{non-rigid} movement model allows an adversarial scheduler to interrupt a robot before it reaches its target, although the robot is guaranteed to make progress by moving at least a minimum distance $\delta > 0$. The non-rigid model is more realistic and significantly complicates Algorithm design and correctness proofs.

Within this framework, the \textit{gathering problem} is one of the most fundamental coordination problems in swarm robotics and distributed computing. Informally, gathering requires all robots in the system to move to a single common location that is not known in advance. Gathering serves as a basic primitive for more complex tasks such as pattern formation, leader election, and cooperative exploration. As a result, it has been extensively studied under various assumptions on visibility, synchrony, memory, communication, multiplicity, and coordinate agreement.

Despite this extensive body of work, many variants of the gathering problem remain challenging, particularly when robots operate under adversarial sensing conditions. Understanding the minimal assumptions under which gathering is solvable is crucial for both theoretical insights and practical system design. This motivates the study of gathering under increasingly weaker models, where robots must coordinate despite severe limitations on visibility, agreement, and motion.

In this paper, we study the gathering problem under this visibility-fault model and investigate the fundamental limits of coordination for autonomous robot swarms. Such visibility faults may arise due to factors such as memory failures, limited sensing capabilities, hardware malfunctions, or temporary environmental interference. 

\subsection{Related Works}

Many geometric coordination problems arise in swarm robotics. Over the last two decades, researchers have studied several fundamental formation problems, including \emph{gathering}, \emph{convergence}, \emph{circle formation}, \emph{pattern formation}, \emph{flocking}, and \emph{plane formation} \cite{agmon2006fault, Cieliebak2012, Cohen2004, flocci2014, Flocchini2008, Fujinaga2010, GervasiP04, YamauchiUKY17}. Among these, the \emph{gathering problem} is one of the most basic and widely studied problems in distributed mobile robotics.

The gathering problem has been explored under many different computational models and fault assumptions \cite{santorobook2}. A key Goal of this research is to identify the minimum capabilities required for robots to successfully gather, considering different levels of synchrony, visibility, memory, communication, and fault tolerance. Below, we summarize the main results according to the scheduling model.
\begin{itemize}
\item \textbf{Unlimited Visibility Model :}
Most existing studies assume that robots have unlimited visibility. We first review results under this assumption.
\begin{itemize}
\item \emph{2-Dimensional Euclidean Plane}

\textbf{FSYNC Model:}
In the fully synchronous (FSYNC) model, the gathering problem can be solved without any additional assumptions \cite{Santoro2012}. Agmon and Peleg proposed a fault-tolerant gathering Algorithm that tolerates up to $\frac{n}{3}-1$ Byzantine faults \cite{agmon2006fault}. Flocchini \emph{et al.} studied gathering and leader election for robots deployed on a closed curve in $\mathbb{R}^2$ and proposed several randomized Algorithms under different assumptions on memory, communication, visibility, orientation, and speed \cite{FlocchiniKKSY19}.

\textbf{SSYNC Model:}
For $n=2$, the gathering problem is known as the \emph{rendezvous problem}. Suzuki and Yamashita showed that rendezvous is impossible without any agreement on local coordinate systems, even with strong multiplicity detection \cite{Suzuki1999}. Later work showed that rendezvous is solvable in the semi-synchronous model when robots are equipped with lights using only two colors \cite{Viglietta13}.  
For $n>2$, Prencipe proved that gathering is impossible without multiplicity detection and without any agreement on local coordinate systems \cite{Prencipe2007}.

The gathering problem has also been studied under the \emph{crash fault} model. Agmon and Peleg proposed the first fault-tolerant gathering Algorithm that tolerates one crash fault with weak multiplicity detection \cite{agmon2006fault}. Bramas and Tixeuil later presented an Algorithm tolerating an arbitrary number of crash faults, assuming strong multiplicity detection \cite{Quentin2014}. Bhagat and Mukhopadhyaya proposed two fault-tolerant Algorithms assuming weak multiplicity detection \cite{BhagatM17}. Pattanayak \emph{et al.} extended these results by allowing guided trajectories, enabling tolerance of an arbitrary number of faults with weak multiplicity detection \cite{PattanayakMRM19}.  
Izumi \emph{et al.} proposed a constant-round randomized gathering Algorithm for oblivious robots with strong local multiplicity detection \cite{Izumi2013}. D\'efaGo \emph{et al.} established impossibility results and randomized solutions under various fault and scheduler models \cite{DefaGoPP20}.

\textbf{ASYNC Model:}
Cieliebak \emph{et al.} solved the gathering problem for $n>2$ asynchronous robots assuming weak multiplicity detection \cite{Cieliebak2012}. Bhagat \emph{et al.} proposed a fault-tolerant gathering Algorithm under obstructed visibility, assuming one-axis agreement and tolerating an arbitrary number of crash faults \cite{Bhagat201650}.  
Bhagat and Mukhopadhyaya studied gathering in environments with polygonal obstacles and provided solutions for non-rotationally symmetric initial configurations \cite{bhagat2017gathering}. Pattanayak \emph{et al.} proposed a randomized rendezvous Algorithm under the ASYNC$_{IC}$ model, assuming a known speed ratio between robots. Cicerone \emph{et al.} studied a variant of the gathering problem involving predefined meeting points and proposed several Algorithms, including one that optimizes a min--max distance criterion \cite{CiceroneSN18}.

The gathering problem has also been investigated for \emph{fat robots} (robots with physical extent). Czyzowicz \emph{et al.} solved the problem for $n=4$ fat robots \cite{czyzowicz2009gathering}, and Agathangelou \emph{et al.} extended the result to an arbitrary number of robots assuming common chirality \cite{Agathangelou2012}.

\item \emph{3-Dimensional Euclidean Space}

The gathering problem and the plane formation problem have also been studied in three-dimensional Euclidean space under different computational models. Bhagat \emph{et al.} studied gathering for opaque robots in $\mathbb{R}^3$ and proposed two Algorithms assuming one-axis agreement: one under the ASYNC model for fault-free robots and another under the SSYNC model tolerating an arbitrary number of crash faults \cite{BhagatCM18}. Several Algorithms have also been proposed for the plane formation problem under different models \cite{uehara2016plane, tomita2017plane, YamauchiUKY17}.
\end{itemize}
\item \textbf{Limited Visibility Model :}
Only a limited number of results exist for the gathering problem under restricted visibility.
\begin{itemize}
\item \emph{2-Dimensional Euclidean Plane}

Flocchini \emph{et al.} studied gathering for asynchronous robots with limited visibility, assuming agreement on both coordinate axes \cite{Flocchini2005}. Chatterjee \emph{et al.} proposed a solution for robots with non-uniform visibility ranges \cite{chatterjee2015gathering}. Later, S.~Chaudhuri \emph{et al.} proposed a finite-time distributed Algorithm assuming one-axis agreement and rigid movements \cite{8554579}.

 \item \emph{3-Dimensional Euclidean Space}

Braun \emph{et al.} studied gathering for synchronized robots in $\mathbb{R}^3$ with limited visibility and provided lower and upper bounds on the number of rounds required, assuming rigid movements \cite{braun2020local}. Di Luna \emph{et al.} proposed a distributed Algorithm in which robots gather on the boundary of a circle using angular visibility \cite{di2025gathering}. Yamauchi \emph{et al.} studied pattern formation for robots with limited visibility in three-dimensional Euclidean space \cite{yamauchi2013pattern}.
\end{itemize}
\item \textbf{Defected View Model:}
Unlike the \emph{limited} and \emph{unlimited visibility} models, the defected view model allows robots to omit the observation of certain robots, irrespective of their distance, due to adversarial defects in perception.
Kim \emph{et al.} were the first to introduce this model~\cite{kim2022gathering} and proposed a gathering algorithm for $N \geq 5$ fully synchronous robots under the adversarial distance-based $(N,K)$ defected view model. They also established impossibility results for the $(3,1)$ model and for a relaxed version of the adversarial $(N, K)$ model.
However, the solvability of gathering under the adversarial $(4,2)$ defected view model remained open~\cite{kim2023gathering}. Very recently, Aliberti \emph{et al.} revisited this problem for robots operating under the \textsc{SSYNC} scheduler and derived two new impossibility results~\cite{11318926}.

  \end{itemize}
  \begin{table*}
\centering
\renewcommand{\arraystretch}{1.3}
\small
\begin{adjustbox}{max width=\textwidth}
\begin{tabular}{|c|c|c|c|c|c|c|}
\hline
\textbf{Scheduler} & \textbf{Observation} & \textbf{Axis-agreement} & \textbf{Movements} & \textbf{N} & \textbf{K} & \textbf{Results for Gathering problem}\\
\hline
FSYNC & Adversarial  & No & \textit{Rigid} & $\geq 5$ & N-2 & Solvable in $\mathcal{O}(1)$ round.\cite{kim2023gathering}\\
\hline
 FSYNC &  Distance based & No & \textit{Rigid} & $\geq 4$  & N-2  &  Solvable in 
 $\mathcal{O}(1)$ round.\cite{kim2023gathering} \\
\hline
Any  & Relaxed & No & \textit{Rigid} & Any & $\leq \left\lfloor \dfrac{N-1}{2} \right\rfloor$ & Unsolvable \cite{kim2023gathering} \\
\hline
SSYNC & Any & No & \textit{Rigid} & Any & $\leq N-2$  & Unsolvable\cite{11318926}\\
\hline
SSYNC & Any & No &  \textit{Rigid} & 4 & 2 & Unsolvable even the distinct gathering \cite{11318926}\\
\hline
\textbf{FSYNC} & {\bf Adversarial} & {\bf NO} & {\bf Non-Rigid} & {\bf 4} & {\bf 2} & {\bf Solved in this paper}\\
\hline
\textbf{ASYNC} & {\bf Adversarial} & {\bf One-axis} & {\bf Non-Rigid} & $\geq 3$ & {\bf Any} & {\bf Solved in this paper}\\
\hline
\end{tabular}
\end{adjustbox}
\caption{Existing and new results for the \textit{gathering} problem under Defected view model.}
\label{tab:min_move_results}
\end{table*}

\subsection
{Motivation and Contributions}

Coordinating autonomous mobile robots is challenging when their sensing abilities are unreliable. In many realistic scenarios, robots may fail to observe all other robots due to obstacles, interference, or faults in perception, or memory failure. To model such situations, we consider the recently introduced \emph{defected view model} \cite{kim2023gathering}, where the set of robots visible to an active robot is chosen adversarially. This significantly increases the difficulty of coordination, as robots must make decisions based on incomplete and dynamically changing information. A natural and fundamental question is whether gathering can still be achieved under such a strong visibility defect using minimal robot capabilities.

 Our first contribution concerns a fundamental and previously unresolved case of the gathering problem under the adversarial $(4,2)$ defected view model posed in \cite{kim2023gathering}. Under the adversarial $(4,2)$ model, each active robot can observe only two of the other three robots. We analyze the difficulties arising from partial and changing observations and propose a distributed gathering algorithm that guarantees gathering in finite time. The proposed algorithm operates in a \textbf{fully synchronized} system and does not assume any additional capabilities such as memory, multiplicity detection, or agreement on coordinate systems. This result establishes the solvability of gathering for the smallest non-trivial instance of the defected view model.

 The other contribution of this paper addresses the general setting with $N$ robots under the adversarial $(N, K)$ defected view model, where the parameter $K$ may vary between $1$ and $N-2$ across different robots within the same execution, and may also vary for the same robot across different executions of an algorithm. We design a finite-time distributed algorithm for \textbf{asynchronous} ($\mathcal{ASYNC}$) robots that works for all $1 \le K \le N-2$. Thus, gathering remains possible even if a robot observes only one of the other robots at each activation. The algorithm assumes only an agreement on the direction and orientation of local coordinate axis of the robots.

Both algorithms operate under the non-rigid movement model, where robot motions may be adversarially interrupted but a minimum progress toward the computed destination is guaranteed. Together, these results demonstrate that deterministic gathering remains achievable despite adversarial visibility and motion constraints, thereby significantly extending existing results in distributed robot coordination.

\section{General Model and Problem Definition}

Let $\mathcal{R} = \{r_1, r_2, \ldots, r_N\}$ denote a set of $N \geq 2$ identical, anonymous, and oblivious robots. 
Robots are modeled as \textbf{punctiform} entities and are initially placed at arbitrary positions in the two-dimensional Euclidean plane $\mathbb{R}^2$. 
Each robot is endowed with basic motorial capabilities, allowing it to move freely in the plane.
\begin{itemize}
\item \textbf{Execution Model}

At any time, a robot can be either \textit{active} or \textit{inactive}. 
When active, a robot executes a standard \textbf{Look--Compute--Move} cycle. 
The system is controlled by an adversarial scheduler, which determines the activation times of robots and the duration of each phase.

Robot motion follows the \textbf{non-rigid movement model}. 
Specifically, after computing a destination point, a robot may be stopped by the scheduler before reaching its target. 
However, it is guaranteed that, in each activation cycle, the robot moves at least a minimum distance $\delta > 0$ toward the computed destination, unless the destination is reached earlier.

\item \textbf{Coordinate Systems and Visibility}

Robots do not share a global coordinate system. 
Instead, each robot $r_i$ is equipped with its own local coordinate system, whose origin coincides with its current position. 
The orientation and scaling of local coordinate systems may differ across robots, unless otherwise specified.

Robots have an \textbf{unlimited visibility range}, meaning that, in the absence of faults, a robot could potentially observe all other robots in the system. 
However, in this work, robot observations are subject to a \textbf{visibility defect}, which limits the set of robots that can be sensed during a Look phase.

\item \textbf{Configurations and Multiplicity}

Multiple robots may occupy the same point in the plane. 
Robots do not possess any multiplicity detection capability. 
Consequently, if several robots are located at the same position, an observing robot perceives that position as occupied by a single robot.

A \textbf{configuration} of the system at time $t$ is defined as the set of distinct points in $\mathbb{R}^2$ currently occupied by robots. 
We denote the configuration at time $t$ by
\[\mathbb{P}(t) = \{p_1(t), p_2(t), \ldots, p_{m}(t)\},\]
where $1 \leq m \leq N$, and each point $p_j(t) \in \mathbb{R}^2$ represents a location occupied by at least one robot. For notational convenience, we sometimes write $p_i$ instead of $p_i(t)$, and $p_i'$ instead of $p_i(t')$. This simplification does not introduce any ambiguity regarding the robot's position.
\\
\end{itemize}
\begin{itemize}
    \item \textbf{(N, K)-Defected View Model:} Consider a system of $N$ small robots. 
During the \emph{Look} phase of an active robot $r_i$, the robot does not observe the complete configuration $\mathbb{P}(t) \setminus \{r_i\}$. 
Instead, an adversary selects an arbitrary subset of at most $K$ robot positions from the set of positions occupied by the other $N-1$ robots, where
\[
1 \leq K \leq N-2.
\]
If the number of distinct robot positions (excluding $r_i$) is less than or equal to $K$, then $r_i$ observes all of them. 
Otherwise, the adversary selects a subset of exactly $K$ positions.

\end{itemize}
\begin{figure}[h]
    \centering
    \includegraphics[width=0.3\textwidth]{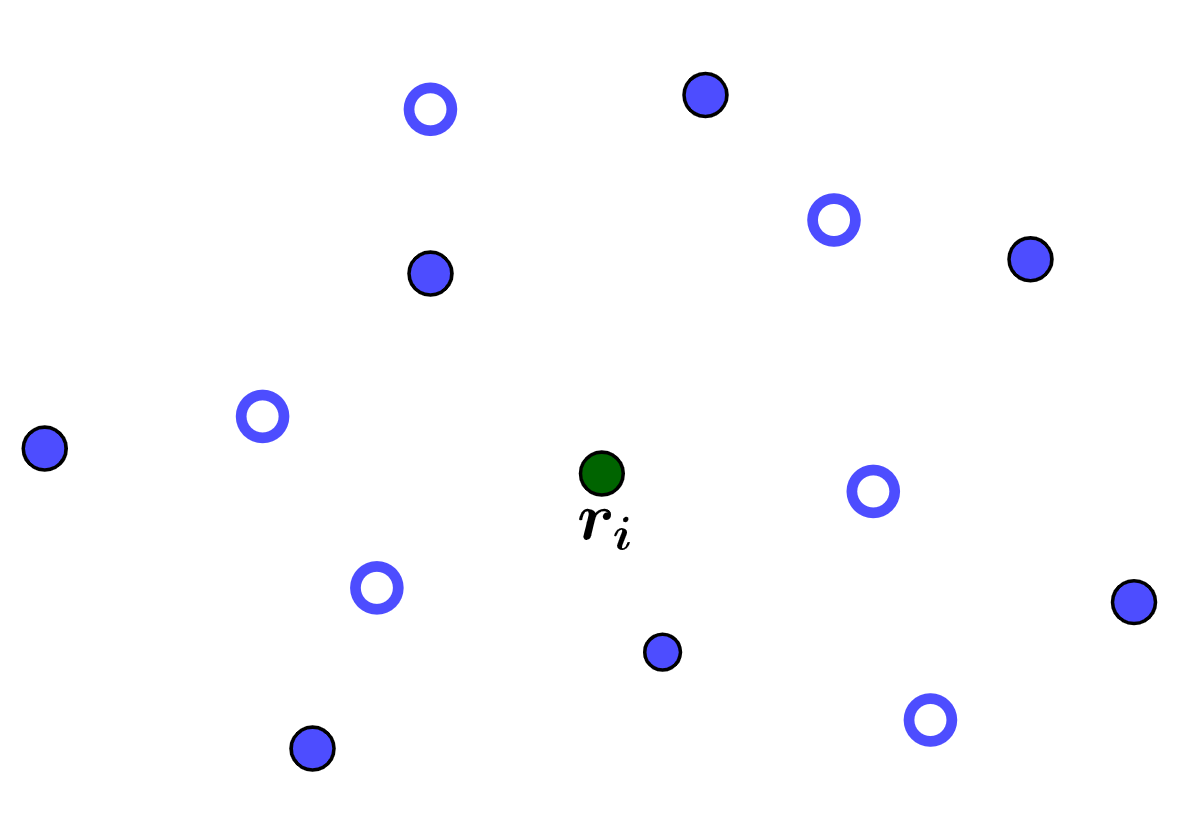}
    \caption{{\it Solid points represent the robots observed by the active robot $r_i$, and }(\textbf{o}) {\it denotes the missed robots by $r_i$.}}
    \label{fig: figure 1.1}
\end{figure}
The subset of observed positions may vary across the activation cycles, and no assumption is made about consistency between successive observations. This model captures severe adversarial sensing limitations, where an active robot has no control over which robots it observes and may perceive a different partial view of the system in each activation cycle.

\section{Adversarial(4,2)-Defected Model}

The adversarial $(4,2)$-defected view model was introduced in the seminal work of Kim \emph{et al.}~\cite{kim2022gathering}, where it was identified as an open problem. While the gathering problem under the \emph{distance-based} $(4,2)$-defected view has recently been solved~\cite{kim2023gathering}, the adversarial variant remained unresolved, even in the fully synchronous setting. In this section, we address this open problem by studying the gathering task for four robots operating under the adversarial $(4,2)$-defected model in a fully synchronous environment, and we propose a distributed algorithm to solve the problem.

Let $\mathcal{R} = \{r_1, r_2, r_3, r_4\}$ denote a set of four anonymous and identical robots deployed in the two-dimensional Euclidean plane. At any time $t$, the position of robot $r_i$ is denoted by $p_i(t)$, and the set of all robot positions at time $t$ is denoted by
\[
\mathbb{P}(t) = \{p_1(t), p_2(t), p_3(t), p_4(t)\}.
\]
The convex hull of $\mathbb{P}(t)$ is denoted by $CH(t)$.

We define the \emph{geometric span} of the configuration at time $t$ as the maximum Euclidean distance between any two points lying on the convex hull:
\[
\mathcal{G}(t) = \max \{\|a - b\| \mid a,b \in CH(t)\}.
\]
The quantity $\mathcal{G}(t)$ captures the overall spatial extent of the robot configuration and will be used as a progress measure in our analysis.

The robots operate in a \emph{fully synchronous} (FSYNC) model. That is, all robots execute the Look–Compute–Move cycle simultaneously and instantaneously in discrete rounds. Each robot is oblivious, has no persistent memory, and does not possess any unique identifiers.

According to the adversarial $(4,2)$-defected view model, during each observation phase, an active robot can observe the positions of at most two other robots among the remaining three. The choice of which robots are observed is controlled by an adversary and may vary arbitrarily across robots and across rounds. The observing robot has no information about which robots are hidden from its view.

For a robot $r_i$, let $\mathcal{O}_i(t)$ denote the set of positions observed by $r_i$ at time $t$. Since a robot always observes its own position, we have
\[
\mathcal{O}_i(t) = \{p_i(t)\} \cup \{p_j(t), p_k(t)\},
\]
where $j,k \in \{1,2,3,4\} \setminus \{i\}$ and $j \neq k$. For example, if robot $r_1$ observes robots $r_2$ and $r_3$, then
\[
\mathcal{O}_1(t) = \{p_1(t), p_2(t), p_3(t)\}.
\]

Robots do not have the ability to detect multiplicity. Consequently, if two or more robots occupy the same location, an observing robot perceives that location as being occupied by exactly one robot. Moreover, when two robots, say $r_1$ and $r_2$, are co-located, an observing robot may sense either $r_1$ or $r_2$, but never both simultaneously.

To illustrate the asymmetry induced by the adversarial observations, consider the following example. Suppose robot $r_3$ observes robots $r_1$ and $r_4$, resulting in
\[
\mathcal{O}_3(t) = \{p_1(t), p_3(t), p_4(t)\},
\]
while robot $r_4$ observes robots $r_2$ and $r_3$, yielding
\[
\mathcal{O}_4(t) = \{p_2(t), p_3(t), p_4(t)\}.
\]
In this scenario, each robot has a different partial view of the same global configuration, and no robot can infer the complete set $\mathbb{P}(t)$ from its local observation alone.

An example of such adversarial observation patterns is illustrated in Figure~\ref{fig: figure 1}.

   \begin{figure}[h]
    \centering
    \includegraphics[width=0.35\textwidth]{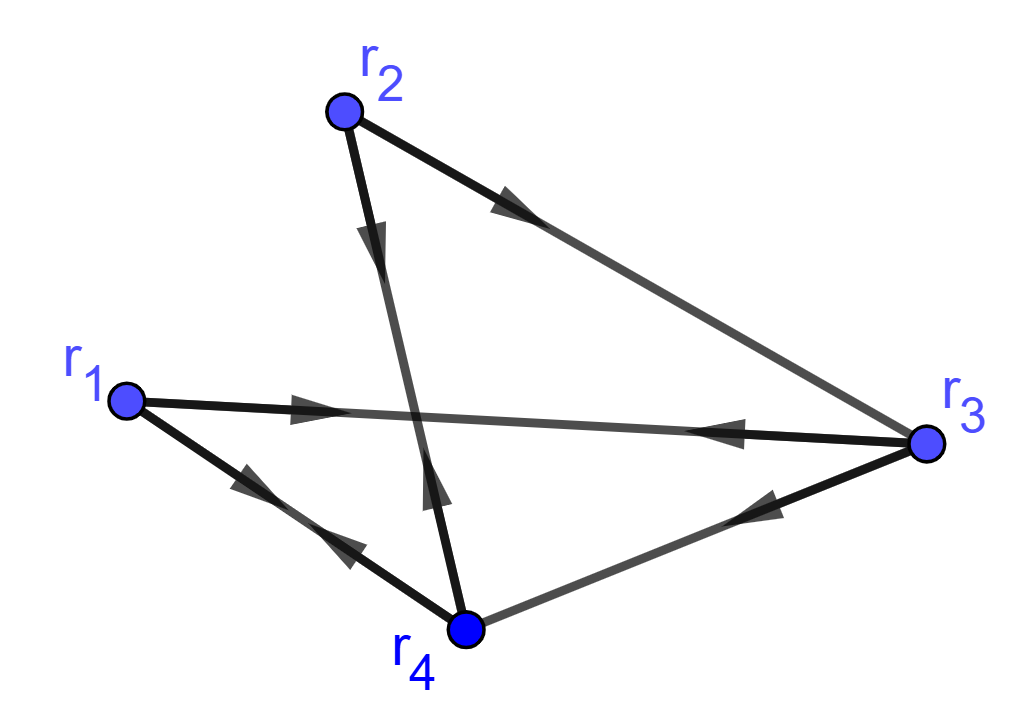}
    \caption{\emph{An edge between two positions does not imply mutual visibility. Arrows indicate the observation relationships among the robots.}}

    \label{fig: figure 1}
\end{figure}
  If robot $r_i$ observes collinear positions and determines that its own position lies strictly between the two visible positions, then $p_i$ is classified as an \emph{internal} position. Otherwise, $p_i$ is classified as an \emph{external} position.

Under this model, we propose a finite-time distributed Algorithm that solves the gathering problem without requiring any additional capabilities, prior agreements, or explicit coordination among the robots. The robots are assumed to move according to a \emph{non-rigid} motion model, meaning that an adversary may interrupt a robot’s movement before it reaches its computed destination, while still ensuring a minimum positive progress.

In order to compute the destination during the execution of the Algorithm, we now introduce several definitions and auxiliary functions that will be used by an active robot based solely on its local observation.
 \\
 \begin{itemize}
 \item {\bf The Vertex and Base of an isosceles triangle:} The {\bf vertex} of the isosceles triangle is defined as the intersection of two equal-length segments. The side, opposite to the vertex, is considered as the base of the triangle as shown in Figure~\ref{fig:  figure2}.
\begin{figure}[h]
    \centering
    \includegraphics[width=0.35\textwidth]{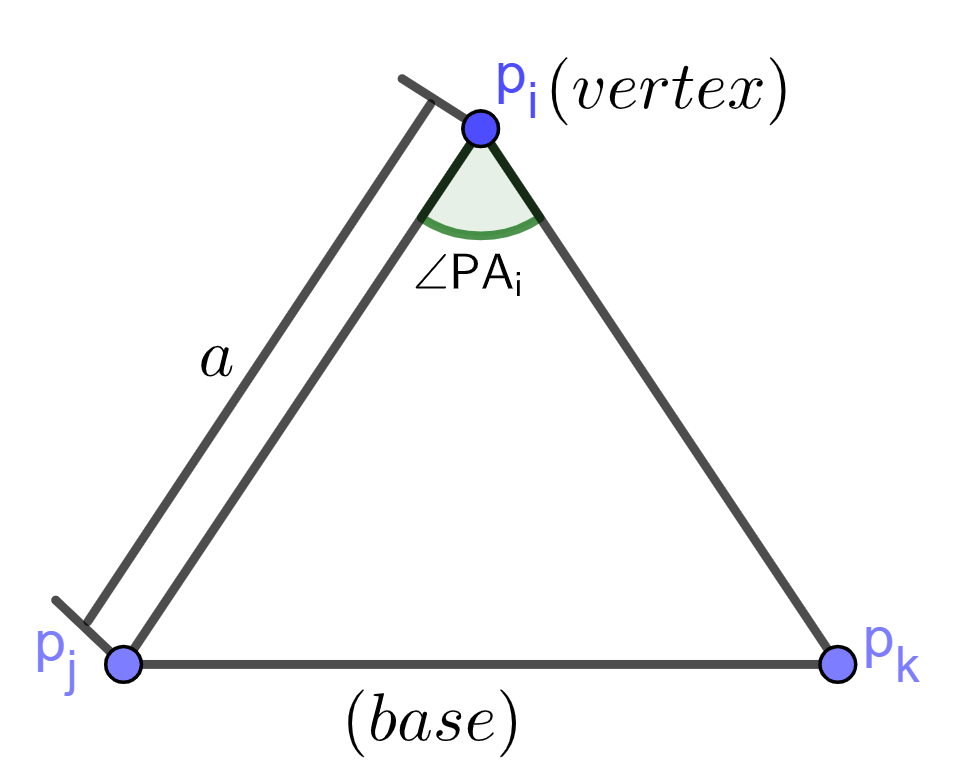}
    \caption{{\it An example of an isosceles triangle.}}
    \label{fig:  figure2}
\end{figure}
     \item \textbf{ComputeMidpoint():}  
This function takes as input the local observation set $\mathcal{O}_i(t)$ and the
current location of robot $r_i$ with respect to the observed positions, and
returns the midpoint of the line segment joining a selected pair of positions in
$\mathcal{O}_i(t)$.

\item \textbf{ComputeTriangle():}  
This function takes the observation set $\mathcal{O}_i(t)$ and robot $r_i$ as inputs, and constructs the triangle formed by the three positions in $\mathcal{O}_i(t)$.

\item \textbf{ComputeLongestLine():}  
Given the observation set $\mathcal{O}_i(t)$, this function returns the longest side of the triangle formed by all the positions in $\mathcal{O}_i(t)$.

\item \textbf{ComputeVertexAngle():}  
This function takes two line segments of equal length as input and returns the angle between them at their common endpoint (vertex angle). For robot $r_i$, this angle is denoted by $\angle PA_i$ as illustrated in Figure~\ref{fig:  figure2}.

 \end{itemize}
\section{Algorithm~\ref{alg:destination}}
\subsection{High-Level Idea}

\paragraph{High-Level Idea.}
In the adversarial $(4,2)$-defected view model, each robot may observe only a partial and adversarially chosen subset of robot positions, with no agreement on coordinates, no multiplicity detection, and non-rigid motion. Therefore, the algorithm is designed to rely solely on local geometric information computed from the observation set of each robot.

In each fully synchronous round, a robot analyzes the geometric configuration
induced by its observation set and classifies it into simple cases, such as a single point, a collinear configuration, or a non-collinear configuration forming a triangle. Based on this classification, the robot deterministically selects a destination defined by a specific geometric feature of the observed configuration, including midpoints, centroids, or midpoints of selected edges.

All computed destinations lie within the convex hull of the observed positions, and robots either move toward such destinations or remain stationary. As a result, robots never move away from the observed configuration, and the overall spatial span of the system does not increase. Symmetric configurations that could prevent progress are handled through carefully designed waiting rules, ensuring that symmetry is eventually broken even under adversarial defected observations.

Despite incomplete and adversarial views, the collective execution of these local
rules guarantees that robots repeatedly reduce their mutual separation whenever
possible. Once all robots reach a configuration in which they observe no other
robot at a distinct position, they correctly infer that gathering has been
achieved and terminate. Hence, Algorithm~\ref{alg:destination} ensures gathering
at a single, initially unknown location in finite time for fully synchronous,
oblivious robots in the adversarial $(4,2)$-defected model.

\subsection{Detailed Description of Algorithm~\ref{alg:destination}}
Depending on the number of observed robot positions and the geometric structure induced by these positions together with its current location, each active robot $r_i$ determines its destination in the Compute phase and moves toward it during
the Move phase. The algorithm classifies the observed configuration into a finite
set of simple geometric cases, such as a single point, a collinear configuration,
or a non-collinear configuration forming a triangle. For each case, a
deterministic destination is computed based solely on locally observable
geometric features. These movement rules ensure consistent progress toward
gathering while preserving correctness under partial and adversarial
observations.

 \begin{itemize}
     \item \textbf{Case-1} $\boldsymbol{|\mathcal{O}_i(t)|=1}$: Suppose robot $r_i$ does not observe any other robot during the \emph{Look} phase at time $t$. Under the adversarial defected-view model, $r_i$ infers that no robot is located at a position distinct from $p_i(t)$ and therefore considers $p_i(t)$ to be the desired gathering point. Consequently, $r_i$ ceases further execution.

       \begin{figure}[h]
    \centering
    \includegraphics[width=0.35\textwidth]{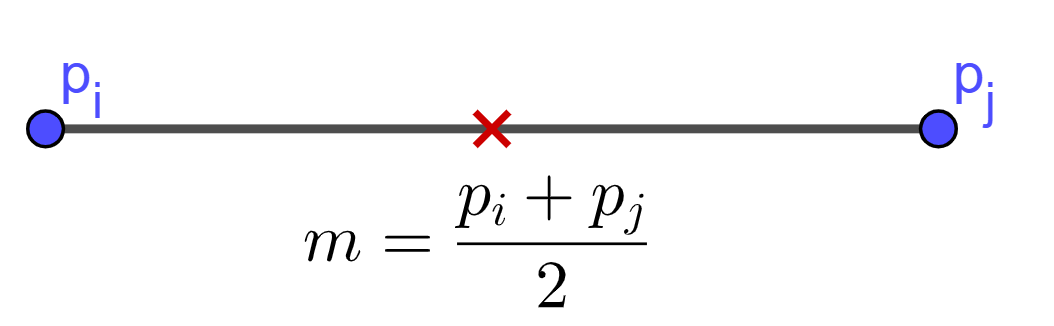}
    \caption{{\it (Illustration of Case-1). 
    Robots in both positions can compute $m$ as midpoint.}}
    \label{fig: figure 3}
\end{figure}
      \item \textbf{Case-2} $\boldsymbol{|\mathcal{O}_i(t)|=2}$: Suppose robot $r_i$ is located at point $p_i(t)$, and observes only one position, say $p_j(t)$ distinct from $p_i(t)$. Using the function \texttt{ComputeMidpoint()}, robot $r_i$ computes the midpoint $m$ of the line segment $\overline{p_i(t)p_j(t)}$ joining its current $p_i(t)$ and $p_j(t)$ the observed position. In the move phase of the computational cycle, $r_i$ moves towards the computed midpoint as illustrated in Figure~\ref{fig: figure 3}. 
    
     \item \textbf{Case-3} $\boldsymbol{|\mathcal{O}_i(t)|=3}$: Suppose robot $r_i$ is located at $p_i(t)$ and observes two distinct robot positions, say $p_j(t)$ and $p_k(t)$, both different from $p_i(t)$. Let $\mathcal{O}_i(t)$ denote the set consisting of these three positions. The points
in $\mathcal{O}_i(t)$ may be collinear or non-collinear. Our algorithm determines the movement of $r_i$ based on the geometric configuration induced by $\mathcal{O}_i(t)$. The corresponding sub-cases are described below.

     \begin{itemize}
         \item \textbf{Case-3.1:} Suppose that the set $\mathcal{O}_i(t)$ induces a {\bf collinear configuration}. First, consider the case in which robot $r_i$ is located at position $p_i(t)$ such that $p_i(t) \in \overline{p_j(t)p_k(t)}$. In this case, using the auxiliary function
\texttt{ComputeMidpoint()}, robot $r_i$ computes the midpoint of the line segment
$\overline{p_j(t)p_k(t)}$. The robot then sets this midpoint as its destination
and performs a non-rigid motion toward it.

         Now consider the case in which the robot position $p_i(t)$ is one of the endpoints
of the collinear configuration, that is, $p_i(t)$ is an extreme point. Without
loss of generality, suppose that $p_k(t)$ is the other extreme point. In this
case, using the function \texttt{ComputeMidpoint()}, robot $r_i$ computes the
midpoint of the line segment $\overline{p_i(t)p_k(t)}$ and moves toward it by
performing a non-rigid motion. An illustration of this case is shown in
Figure~\ref{fig: figure 4}.

         \begin{figure}[h]
    \centering
    \includegraphics[width=0.3\textwidth]{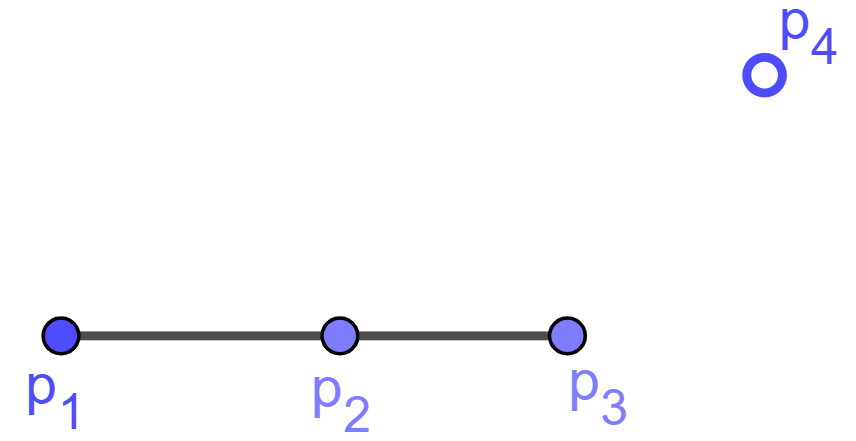}
    \caption{{\it (An illustration of case-3.1). Any robot located in positions $p_1$, $p_2$, and $p_3$ misses the robot at $p_4$ then $\mathcal{O}_i$ will be collinear for $1\leq i \leq 3.$}}
    \label{fig: figure 4}
\end{figure}
         
         \item \textbf{Case-3.2:} Suppose that the configuration induced by the set $\mathcal{O}_i(t)$ is
\textbf{non-collinear}. In this case, the convex hull of the three positions forms
a triangle. This triangle may satisfy certain special geometric properties. Robot
$r_i$ computes the triangle associated with $\mathcal{O}_i(t)$ using the auxiliary
function \texttt{ComputeTriangle()}, which also identifies the type of the
triangle. Based on these geometric properties, the algorithm enables the robot to
move correctly even in the presence of symmetry in the configuration, while
continuously reducing the span of the configuration. The description of the sub-cases is as follows:
 \begin{figure}[h]
    \centering
    \includegraphics[scale=0.7]{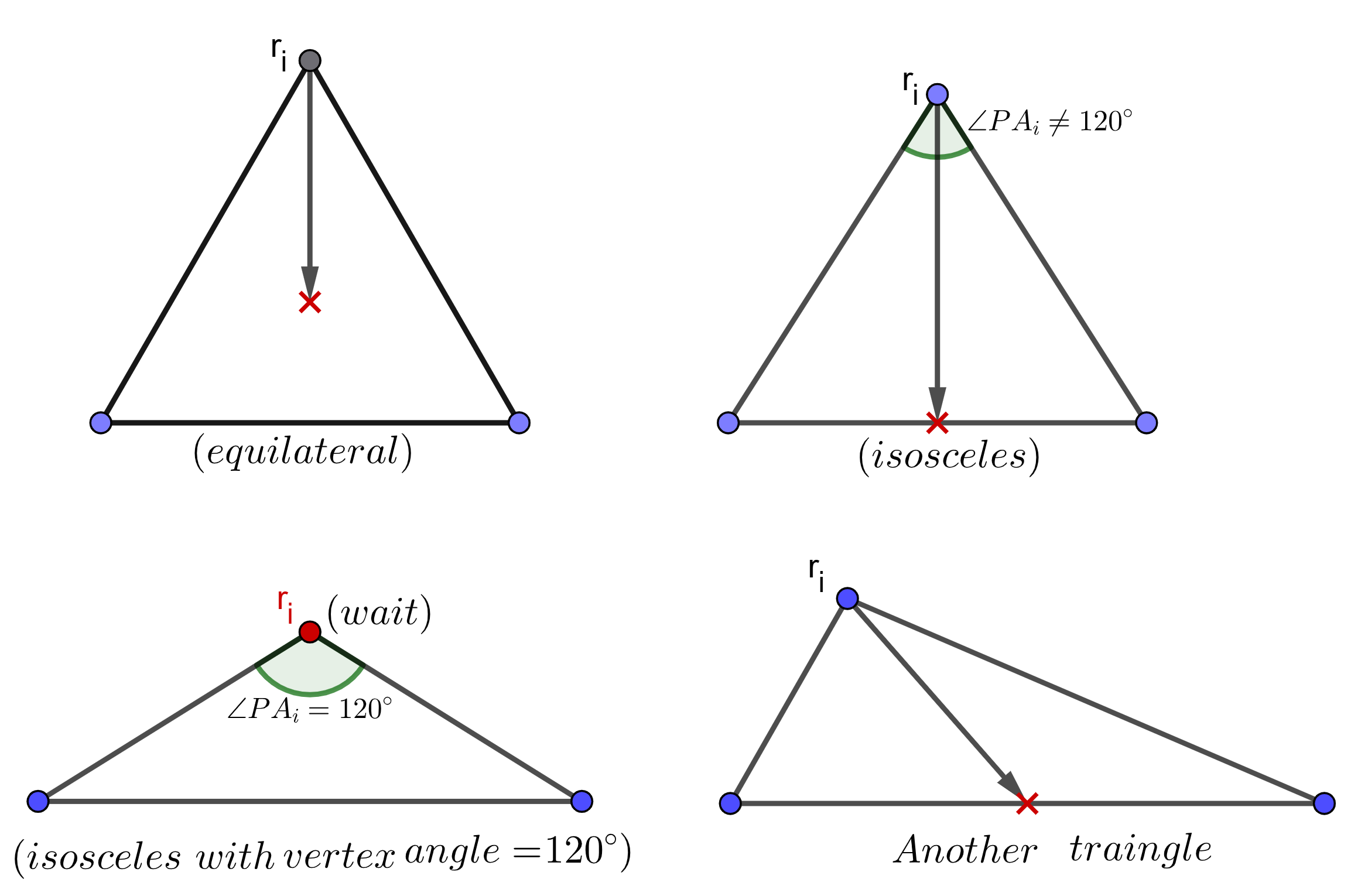}
    \caption{{\it (An illustration of case-3.2). If $\mathcal{O}_i(t)$ is non-collinear, every active robot $r_i$ associated with $\mathcal{O}_i(t)$ may encounter one of the following configuration types.}}
    \label{fig: figure 5}
\end{figure}
         \begin{itemize}
              \item \textbf{Case-3.2.1:} Suppose the set $\mathcal{O}_i(t)$ induces an \textbf{equilateral triangle}. In this
case, robot $r_i$ computes the centroid of the equilateral triangle as its next
destination and initiates a non-rigid motion toward it(see Figure~\ref{fig: figure 5}).

         \item \textbf{Case-3.2.2:} Now consider the case in which the set $\mathcal{O}_i(t)$ forms an \textbf{isosceles triangle}. Suppose that robot $r_i$ is located at the vertex of the triangle, that is, $p_i(t)$ corresponds to the vertex point. In this case, robot $r_i$ computes the vertex angle $\angle PA_i$ (as illustrated in
Figure~\ref{fig: figure2}) using the auxiliary function
\texttt{ComputeVertexAngle()}.

If the vertex angle satisfies $\angle PA_i = 120^\circ$, then robot $r_i$ waits.
This waiting strategy enables the algorithm to break the symmetry induced by the
adversarial defected observation when the global configuration of robots matches
the configuration shown in Figure~\ref{fig: figure 10}. If
$\angle PA_i \neq 120^\circ$, then robot $r_i$ computes the midpoint of the base of
the isosceles triangle, sets it as its destination, and performs a non-rigid
motion toward the computed midpoint during the Move phase.

         If robot $r_i$ is not located at the vertex of the isosceles triangle, that is,
$p_i(t)$ lies at one of the base vertices, then robot $r_i$ computes the midpoint
of the base of the triangle and performs a non-rigid motion toward it.

        \item \textbf{Case-3.2.3:} Suppose the set $\mathcal{O}_i(t)$ induces a triangle that is neither equilateral nor isosceles. In this case, robot $r_i$, located at $p_i(t)$, first computes the longest side of the triangle using the auxiliary function
\texttt{ComputeLongestLine()}. Let $\mathcal{L}_i$ be the longest line. It then computes the midpoint of this longest side and performs a non-rigid motion toward the computed midpoint.

     \end{itemize}
         \end{itemize}
     \end{itemize}
    Each active robot executes Algorithm~\ref{alg:destination} independently, and the
collective execution guarantees that all robots gather at a single,
non-predefined location in finite time.




\begin{algorithm}[!htb]
\caption{\texttt{ComputeDestination($r_i$)}}
\label{alg:destination}
\begin{algorithmic}[1]

\Require Observation set $\mathcal{O}_i$, current position $p_i$
\Ensure Destination point $d_i$

\If{$|\mathcal{O}_i| = 1$}
    \State $d_i \gets p_i$

\ElsIf{$|\mathcal{O}_i| = 2$}
    \State Let $p_j$ be the other point in $\mathcal{O}_i$
    \State $d_i \gets \texttt{ComputeMidpoint}(p_i, p_j)$

\ElsIf{$|\mathcal{O}_i| = 3$}
    \State Let $p_j, p_k$ be the other two points in $\mathcal{O}_i$
    
    \If{$\{p_i, p_j, p_k\}$ is collinear}
        \State Let $q_{\min}, q_{\max}$ be the extremal points of $\{p_i, p_j, p_k\}$
        \If{$p_i \in \{q_{\min}, q_{\max}\}$}
            \State Let $p_{\text{far}}$ be the point farthest from $p_i$ in $\{p_j, p_k\}$
            \State $d_i \gets \texttt{ComputeMidpoint}(p_i, p_{\text{far}})$
        \Else
            \State $d_i \gets \texttt{ComputeMidpoint}(q_{\min}, q_{\max})$
        \EndIf

    \Else
        \If{$\{p_i, p_j, p_k\}$ forms an equilateral triangle}
            \State $d_i \gets \texttt{ComputeCentroid}(p_i, p_j, p_k)$

        \ElsIf{$\{p_i, p_j, p_k\}$ forms an isosceles triangle}
            \State Let $v$ be the vertex of the isosceles triangle
            \State Let $\{b_1, b_2\}$ be the endpoints of the base opposite $v$
            \If{$p_i = v$ and $\angle b_1 v b_2 = 120^\circ$}
                \State $d_i \gets p_i$
            \Else
                \State $d_i \gets \texttt{ComputeMidpoint}(b_1, b_2)$
            \EndIf

        \Else
            \State Let $\{s_1, s_2\}$ be the endpoints of the longest side of $\triangle p_i p_j p_k$
            \State $d_i \gets \texttt{ComputeMidpoint}(s_1, s_2)$
        \EndIf
    \EndIf

\Else
    \State $d_i \gets p_i$
\EndIf

\State \Return $d_i$

\end{algorithmic}
\end{algorithm}

   \subsection{Correctness}
In this part of the paper, we prove the main theorem establishing the correctness of Algorithm~\ref{alg:destination}. Specifically, we show that, under the adversarial $(4,2)$ defected view model, the proposed Algorithm guarantees that all robots gather at a single point within finite time. The proof demonstrates that gathering is achieved by synchronous robots despite defected views and non-rigid motion.

Our analysis proceeds by identifying key invariants preserved throughout the execution and by showing monotonic progress toward a common location. We further establish that the Algorithm does not rely on any additional capabilities, such as global knowledge, explicit communication, multiplicity detection, or extra agreement assumptions beyond those of the model. Hence, Algorithm~\ref{alg:destination} provides a correct and self-contained solution to the gathering problem under the given adversarial setting.

     \begin{lemma} {\it If robots are collinear, then at least two robots compute the same midpoint and all the collinear robots gather at a single point in a finite number of rounds.}\label{1.1}
     \begin{proof} Suppose there are four collinear positions, say $p_1, p_2, p_3$ and $p_4$. Let $p_1$ and $p_4$ be located at the corner positions as shown in Figure~\ref{fig: figure 6}. Under the (4,2)-defected model, each robot adversarially chooses two positions other than its own. According to their observation,  
     
     \textbf{Case 1.} If all the robots find themselves at an external position, then definitely $\mathcal{O}_1= \{p_1,p_2,p_3\}$ or $ \{p_1, p_2,p_4\} $ or $ \{p_1,p_3,p_4\}$, and similarly, $\mathcal{O}_4= \{p_1,p_2,p_4\}$ or $ \{p_1, p_3,p_4\} $ or $ \{p_2,p_3,p_4\}$, but the set $\mathcal{O}_2= \{p_2,p_3,p_4\}$, and $\mathcal{O}_3= \{p_1,p_2,p_3\}$ in this scenario, $r_1$ either computes midpoint of the segment $\overline{p_1 p_3}$ or midpoint of $\overline{p_1 p_4}$. $r_4$ either computes midpoint of segment $\overline{p_1 p_4}$ or midpoint of $\overline{p_2 p_4}$. Robots $r_2$ and $r_3$  computes midpoint of the segments $\overline{p_2 p_4}$ and $\overline{p_1 p_3}$ respectively. Thus, if $r_1$ and $r_4$ compute the same midpoint, the proof is complete. If both compute different midpoints, then either $r_1$ or $r_4$ computes the midpoint that is the same as $r_2$ or $r_3$.
         \begin{figure}[h]
    \centering
    \includegraphics[width=0.3\textwidth]{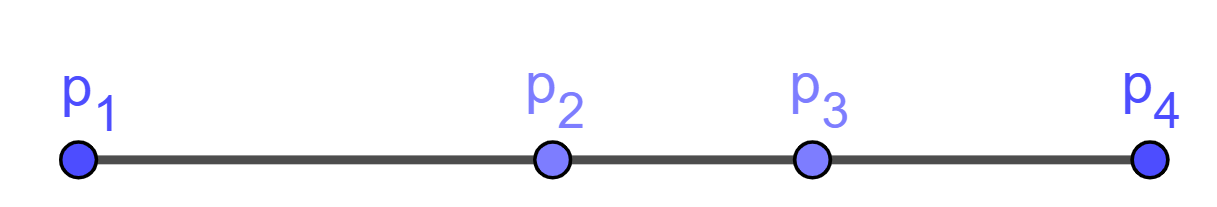}
    \caption{$\mathbb{P}(t)$ is collinear.}
    \label{fig: figure 6}
\end{figure}
     
     \textbf{Case 2.} If only one non-corner robot say $r_3$ finds itself at an {\it internal} position at $p_3$, then $r_3$ computes either midpoint of $\overline{p_1 p_4}$ or midpoint of $\overline{p_2 p_4}$, and $r_1$, $r_2$, $r_4$ compute midpoint same as Case 1. Again, for all possible choices of midpoints, we can conclude that at least two robots compute the same midpoint.\\

     \textbf{Case 3.}  If two robots find themselves at an {\it internal} position, then $r_2$ computes either the midpoint of $\overline{p_1 p_3}$ or the midpoint of $\overline{p_1 p_4}$, $r_3$ computes the midpoint as in Case 2, and $r_1$ and $r_4$ compute as in Case 1. From all possible computed midpoints, we can conclude that at least two robots compute the same midpoint.\\

     As no robot moves outside the segment $\overline{p_1 p_4}$ and robots $r_1$, $r_4$ always move inside, so the length of the segment monotonically decreases and becomes less than $2\delta$. In this situation, from the above cases, we can say that every robot computes the midpoint accordingly and reaches it in a round. Now, there remain only 3 positions. In the next round, suppose $p_1'$, $p_2'$, and $p_3'$ are three collinear positions, then according to the proposed Algorithm, all the robots located at these positions compute the same midpoint and reach it. In this way, the collinear robots gather at a single location in a finite number of rounds.
     \end{proof}
     \end{lemma}
    
 \begin{lemma} {\it A convex hull of $4$ distinct robot positions that does not contain any isosceles or equilateral triangle, then such a configuration has at most 3 longest lines. } \label{1.2}
 \begin{proof} Let $p_1,p_2,p_3$, and $p_4$ be four robot positions in the convex hull $CH()$, and from these positions, a total of 6 segments can be drawn as in Figure~\ref{fig: figure 7}. Every observed set has at most $3$ distinct elements. According to the observed set, $4_{C_3}$ distinct triangles can be formed. Let these triangles be $\Delta_{1}=\Delta_{123}$, $\Delta_{2}=\Delta_{124}$, $\Delta_{3}=\Delta_{134}$, $\Delta_{4}=\Delta_{234}$.
  \begin{figure}[h]
    \centering
    \includegraphics[scale=0.8]{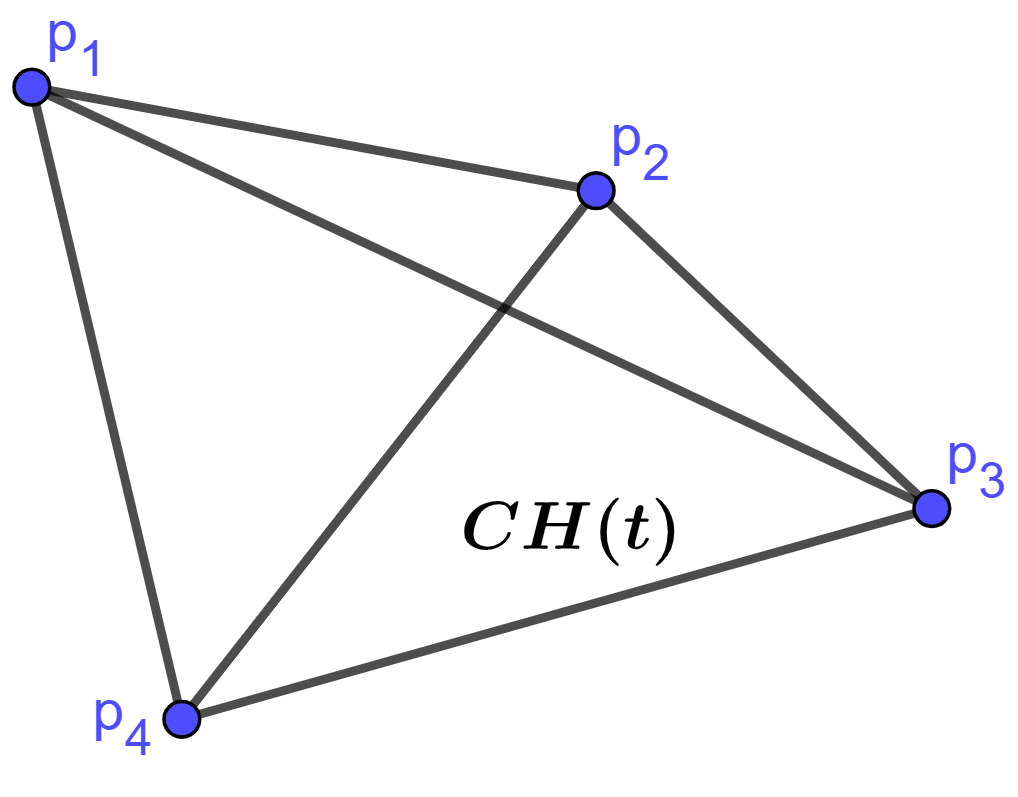}
    \caption{{\it (An illustration of Lemma~\ref{1.2}) The convex hull $CH(t)$ contains neither equilateral nor isosceles triangle.}}
    \label{fig: figure 7}
\end{figure}
 
 For contradiction, suppose there are 4 longest lines, i.e., each triangle has a unique longest line. WLOG, suppose $\mathcal{L}_1=\overline{p_1p_3}$, $\mathcal{L}_2=\overline{p_1p_4}$, $\mathcal{L}_3=\overline{p_3p_4}$, $\mathcal{L}_4=\overline{p_2p_4}$ are longest lines. $\mathcal{L}
 _{k}$ is the longest line of the triangle $\Delta _k$. $\mathcal{L}=\{\mathcal{L}_1, \mathcal{L}_2, \mathcal{L}_3, \mathcal{L}_4\}$ be the set of longest lines.  By a key combinatorial fact, every edge appears in exactly two of the four triangles. If $\mathcal{L}_k$ is the longest line of the triangle $\Delta_k$. Let $\Delta_{k'}$ be any other triangle containing edge $\mathcal{L}_k$, then $\exists \mathcal{L}_{k'}$ such that $|\mathcal{L}_{k'}|>|\mathcal{L}_{k}|$. From all the triangles, we can conclude that, 
 \begin{equation}\label{e-1.1}
     |\mathcal{L}_{1}|<|\mathcal{L}_{2}|<|\mathcal{L}_{3}|<|\mathcal{L}_{4}|<|\mathcal{L}_1|
 \end{equation}

     equation~\ref{e-1.1} implies,  $|\mathcal{L}_{1}|<|\mathcal{L}_{1}|$ which is impossible. So our assumption is wrong; there must be at most 3 longest lines of distinct size.
     \end{proof}
     \end{lemma}
     
     \begin{lemma} {\it If all the robots are single and located at the vertex of the convex hull, then at least two robots compute the midpoint on the longest diagonal line.}\label{1.3}
\begin{proof} Let all 4 robots be located at distinct positions, and their observed set does not form an isosceles or equilateral triangle, then by the pigeonhole principle, every edge appears in exactly two of the four triangles, and there must be a longest diagonal line in two triangles so at least two robots compute the midpoint on the longest diagonal line. 
\end{proof}
\end{lemma}
 
\begin{lemma} {\it If all the robots are single and located on the convex hull, then at least two robots compute the midpoint on the same longest line.}\label{1.4}
\begin{proof} By a similar argument to Lemma~\ref{1.3}, we can prove this Lemma. 
\end{proof}
\end{lemma}

 \begin{lemma} {\it The geometric span $\mathcal{G}(t)$ of all the robot positions monotonically decreases,i.e., for any $t'>t$, the span $\mathcal{G}(t)\leq \mathcal{G}(t')$.} \label{1.5}
\begin{proof} In Lemma~\ref{1.1}, we have seen that the length of the line monotonically decreases. Now, suppose robots are not collinear, then according to the proposed Algorithm, robots either move towards the midpoint on the edges of the convex hull or move inside the hull. Let $CH(t)$ be the convex hull of all four positions at time $t$. Let $\overline{p_1(t)p_2(t)}$ be an edge of the convex hull. First, suppose that the robot at $p_1(t)$ moves inside the convex hull. Let it stop at some position, say $p_1(t')$, which is inside $CH(t)$. Similarly, there is another vertex, say $p_3(t)$, then \begin{equation}\label{e-1}
     ||p_1(t')-p_3(t)||\leq || p_1(t)-p_3(t)||
\end{equation}

 \begin{figure}[h]
    \centering
    \includegraphics[scale=0.8]{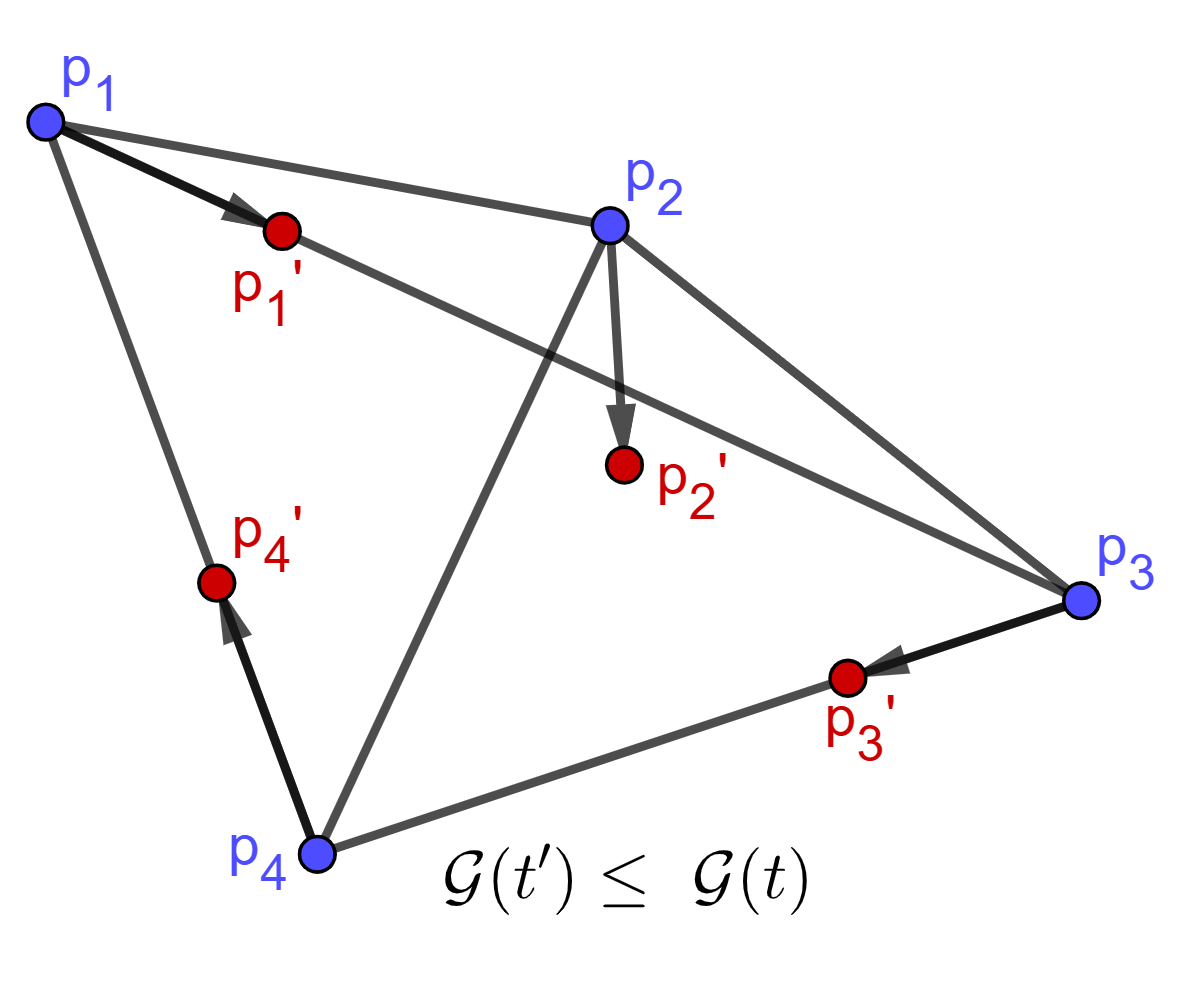}
    \caption{{\it An example showing the non-rigid movement of robots towards the respective destination points.}}
    \label{fig: figure 8}
\end{figure}
In another scenario, when the robot moves along the outer edge of the convex hull, then if $p_1(t')$ lies on the edge of $CH(t)$ between the vertices $p_1(t)$ and $p_2(t)$, then by convexity, we have that $p_1(t') = \lambda p_1(t) + (1-\lambda) p_2(t)$, for $0<\lambda <1$, Then for any vertex $p_3(t)$

{\footnotesize
\begin{multline}\label{e-2}
     ||p_1(t')-p_3(t)||\leq (1-\lambda||p_1(t)-p_3(t')||+\lambda ||p_2(t)-p_3(t')||\\< max(||p_1(t)- p_3(t')||,||p_2(t)- p_3(t')||) 
\end{multline}
}

 From the equations~(\ref{e-1}), (\ref{e-2}), we can conclude that if any of the two robot positions is involved in the pair defining geometric span, then after the motion, we have $\mathcal{G}(t')< \mathcal{G}(t)$.
 As no robot in the hull waits indefinitely, and there exists at least two positions defining the geometrical span, we can conclude that the geometrical span of the robot positions is monotonically decreasing, as in Figure~\ref{fig: figure 8}. 
 \end{proof}
 \end{lemma}

\begin{lemma} {\it  Assume that all the robots are single. If the distance between every robot position is less than or equal to $\ delta$, then gathering can be achieved in at most three rounds.} \label{1.6}
\begin{proof}  According to Lemma~\ref{1.5}, any configuration of robot positions monotonically decreases and converges to a configuration such that $max||x-y||\leq \delta$, for all $x,y \in CH$. This means that every robot reaches the computed destination in the same cycle. Now, we will discuss all possible scenarios.
\textbf{If there are three longest lines:} From the figure, if all the robots are at positions $p_1$, $p_2$, $p_3$, $p_4$, compute the same longest line, then they gather at the midpoint of the longest line. In the second scenario, suppose $r_1$ and $r_2$ observe the triangle $\Delta_{124}$, and suppose $r_3$ and $r_4$ observe the triangle $\Delta_{134}$, then they reach the computed midpoints on the segments $\overline{p_2p_4}$ and $\overline{p_1p_3}$ respectively. In the next round, there will be only two positions, and robots in each position compute the midpoint of the segment and gather at a single position. Now, suppose every robot has a different view under an adversarial $(4,2)$-defected view model, then according to Lemma~\ref{1.3}, at least two robots compute the midpoint for the same longest line. WLOG, suppose $r_1$ and $r_2$ compute the midpoint for the same line segment, say $\overline{p_2p_4}$, and $r_3$ and $r_4$ compute the midpoints for lines say  $\overline{p_1p_3}$ and $\overline{p_3p_4}$ and moves to it, in next round all the robots observes the same triangle and compute the same destination and gathering will be achieved. In all these scenarios, gathering will be achieved in at most 2 rounds. 

\textbf{If there are four longest lines:} According to Lemma~\ref{1.2}, there must be either an isosceles or equilateral triangle. Now, according to Lemma 16 discussed in \cite{kim2023gathering}, that type of configuration can be gathered in two rounds.   

\textbf{If one robot is located inside an isosceles or an equilateral triangle:} First, suppose position $p_4$ is located inside an isosceles triangle $\Delta p_1p_2p_3$(with vertex angle $\angle PA_1\neq 120^\circ$) from Figure~\ref{fig: figure 9}, then if $r_1$ observes $p_2, p_3$, then moves to the midpoint of the segment $\overline{p_2p_3}$(base). Suppose $r_2$(resp $r_3$) observes $p_1, p_4$, then moves to the midpoint of $\overline{p_1p_2}$(resp.$\overline{p_1p_3}$), and when $r_4$ moves, it accompanies one of them, and in the next round, gathering can be achieved.
 \begin{figure}[h]
    \centering
    \includegraphics[width=0.45\textwidth]{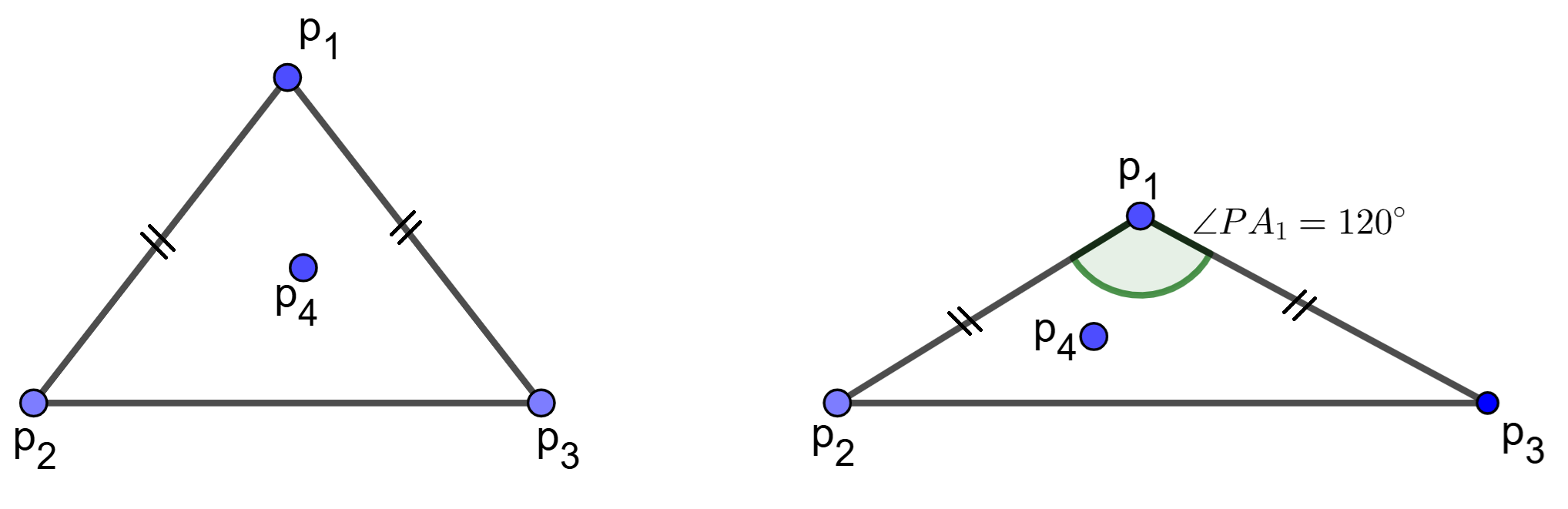}
    \caption{{\it Left one is a normal isosceles, and the triangle at the right is an isosceles triangle with vertex angle of $120^\circ$.}}
    \label{fig: figure 9}
\end{figure}
Now suppose position $p_4$ is located inside an isosceles triangle $\Delta p_1p_2p_3$(with vertex angle $\angle PA_1= 120^\circ$) so $r_1$ can not move, then if $r_4$ observes the same as either $r_2 $ or $r_3$ there will be an accompanied position and in next round all the robots looks the isosceles triangle and move to the midpoint of the base except $r_1$, and in the third round $r_1 $ and all other robot moves to the midpoint of the remaining single line segment. Similarly, we can see for an Equilateral triangle where $p_4$ is not located at the center of the triangle. 
 \begin{figure}[h]
    \centering
    \includegraphics[width=0.35\textwidth]{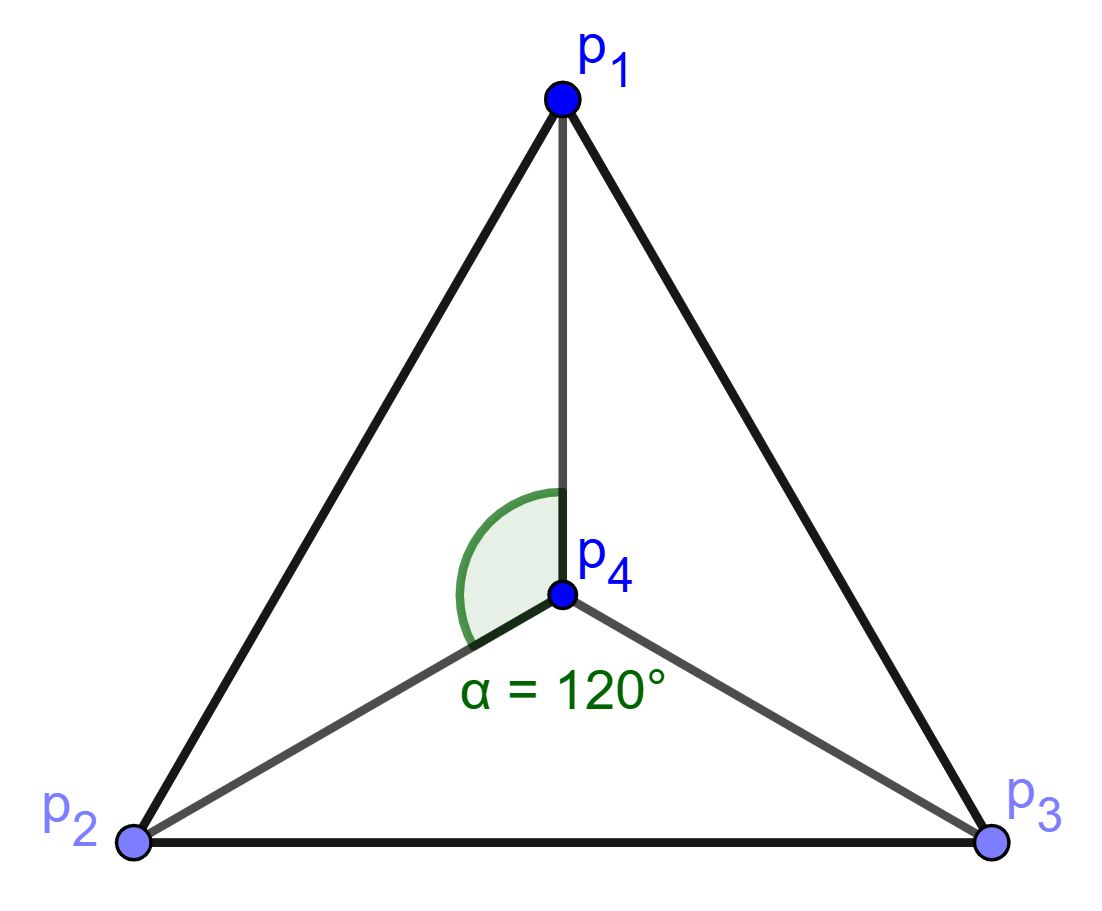}
    \caption{{\it (An illustration of case-3.2) A configuration where three robots are located at the vertices of the equilateral triangle and the remaining one is located at the centroid of the triangle.}}
    \label{fig: figure 10}
\end{figure}

Consider the case where $p_4$ is located at the center of an equilateral triangle $\Delta p_1p_2p_3$  as shown in Figure~\ref{fig: figure 10}. Definitely, $\angle p_1p_4p_3=\angle p_2p_4p_3=\angle p_1p_4p_2=120^\circ$. Suppose if $r_1$ observes $p_2$ and $p_3$ i.e., equilateral triangle $\Delta_{123}$ and $r_2, r_3$ observes $p_1, p_4$, and $r_4$ observes $p_2, p_3$.  Clearly, $\Delta_{124}, \Delta_{134}$ and $\Delta_{234}$ is isosceles triangle with vertex $p_4$ and $\angle PA_4 =120^\circ$. In this scenario, $r_1$ moves to $p_4$ and $r_2$(resp.$r_3$) moves to the midpoint of the base $\overline{p_1p_2}$(resp.$\overline{p_1p_3}$), and $r_4$ waits. Suppose the new position of $r_2$ and $r_3$ is $p_2'$, $p_3'$ respectively. $\Delta_{12'3'}$ is again an isosceles triangle with vertex $p_1$ and $\angle PA_1=120^\circ$ so robots $r_4$ and $r_1$ waits and $r_2, r_3$ moves to the base of $\Delta_{12'3'}$ in second round. In the third round, robots in both positions move to the midpoint and gather at a single location. 
  \end{proof}
  \end{lemma}

 From Lemmas~\ref{1.1},~\ref{1.5}, and ~\ref{1.6}, the following theorem holds.
 
         \begin{theorem}
Algorithm~\ref{alg:destination} gathers all synchronized robots in finite time
under the adversarial $(4,2)$-defected view model.
\end{theorem}

     %

  \section{Gathering under the Adversarial $(N,K)$-Defected View Model}
In this section, we study the gathering problem for a set $\mathcal{R} = \{r_1, r_2, \dots, r_N\}$ of $N \geq 3$ punctiform robots operating asynchronously in the two-dimensional Euclidean plane. Each robot independently executes the classical Look--Compute--Move cycle under a \textbf{non-rigid motion model}, in which movement toward a computed destination may be interrupted by an adversary, provided that a minimum positive progress is always guaranteed.

The robots share a common agreement on the direction of the $Y$-axis in their local coordinate systems, which provides a consistent notion of the north--south direction. However, they do not have any \textbf{multiplicity detection} capability, and hence multiple robots occupying the same position are perceived as a single robot.

A key feature of the model is the \textbf{adversarial defective view}. During the Look phase, a robot may fail to observe up to $N-K$ other robots, chosen adversarially. As a result, at most $K$ robot positions are visible in any observation, including the robot’s own position. We denote by $\mathcal{O}$ the set of robot positions observed during a Look phase.

Let $\mathbb{P}(t) = \{p_1(t), p_2(t), \dots, p_m(t)\}$ denote the set of all $m \leq N$ distinct robot positions at time $t$, where each position may be occupied by one or more robots. In this model, we assume that robots share a common understanding of the orientation and direction of the $Y$-axis. The positive direction of the $Y$-axis is considered as north, while the negative direction is considered as south. This shared agreement allows us to order the robot positions by horizontal lines, based on their $y$-coordinates. The set of such horizontal lines is denoted by
\[
\mathcal{L}(t) = \{L_1(t), L_2(t), \dots, L_b(t)\},
\]
where $|\mathcal{L}(t)| = b$ is the total number of distinct horizontal lines. The set is an ordered sequence from north to south direction,i.e., $L_1(t)$ is the north-most line, also known as the topmost line, and $L_k(t)$ denotes the $k^{th}$ line from the north. The cardinality of $L_k(t)$ is denoted by $|L_k(t)|$, which represents the number of robot positions on the $k^{th}$ horizontal line at any time $t$. Figure~\ref{fig: figure 11}(a) illustrates an example of this arrangement.

     \begin{figure}[h]
    \centering
    \includegraphics[width=0.45\textwidth]{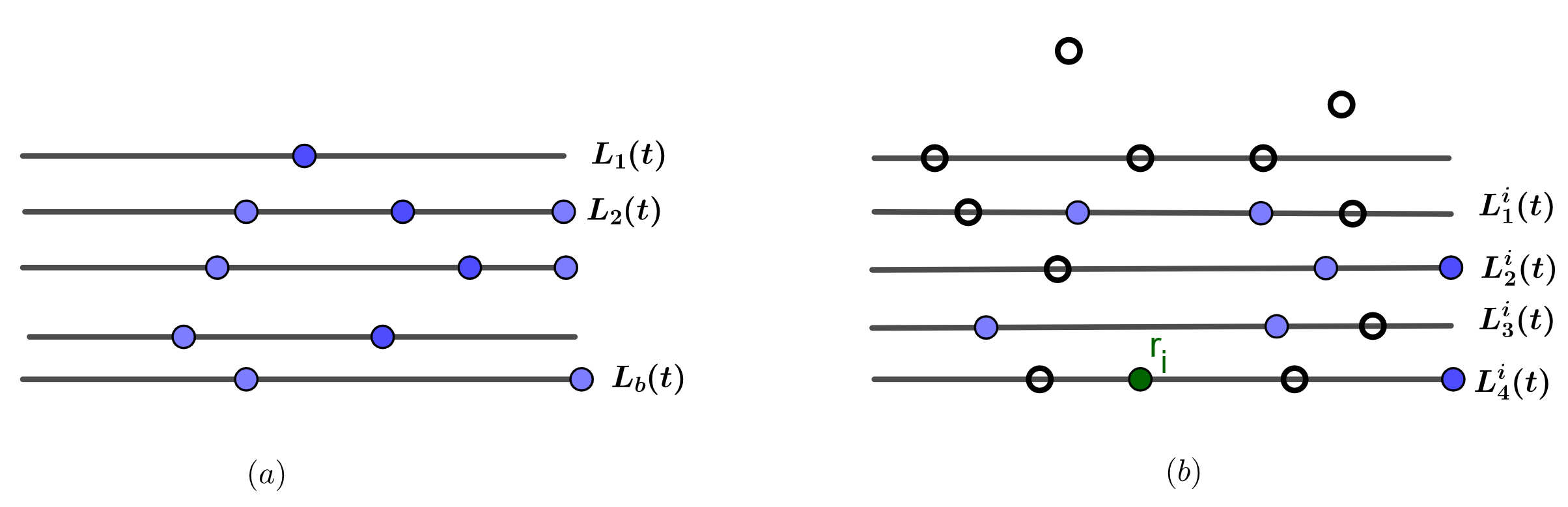}
    \caption{{\it (a). shows the ordering of the robots in $\mathbb{P}(t)$, and (b). shows the ordering of the robots in $\mathcal{O}_i(t)$(the robots observed by robot $r_i$ at time $t$).}}
    \label{fig: figure 11}
\end{figure}
     
    For an individual robot $r_i$ located at position $p_i(t)$, the observed set of positions $\mathcal{O}_i(t)$ (including its own position) can similarly be arranged into a set of horizontal lines, denoted by $\mathcal{L}^i(t)$. These lines are ordered from north to south, where $L_k^i(t)$ represents the $k$-th horizontal line in this ordering, and $|L_k^i(t)|$ denotes the number of robot positions on $L_k^i(t)$ that are visible to $r_i$, as illustrated in Figure~\ref{fig: figure 11}(b).

Due to the defect in the robot’s vision, the observed structure may differ from the actual global configuration. In particular, if $r_i$ does not observe any position above its own, it concludes that it lies on the topmost observed line $L_1^i(t)$ at time $t$. Similarly, if some corner positions are not visible, the robot may incorrectly perceive itself as being at a corner position of the observed line.

If the observed set $\mathcal{O}_i(t)$ is collinear, then $r_i$ further classifies its position based on the relative ordering of the visible positions. If $r_i$ finds its position strictly between the two other visible positions, it defines itself as being at an \emph{internal} position; otherwise, it defines itself as being at an \emph{external} position.

We now introduce several functions and notations that will be used to design Algorithm~\ref{alg:movement}.
 \begin{figure}[h]
    \centering
    \includegraphics[width=0.5\textwidth]{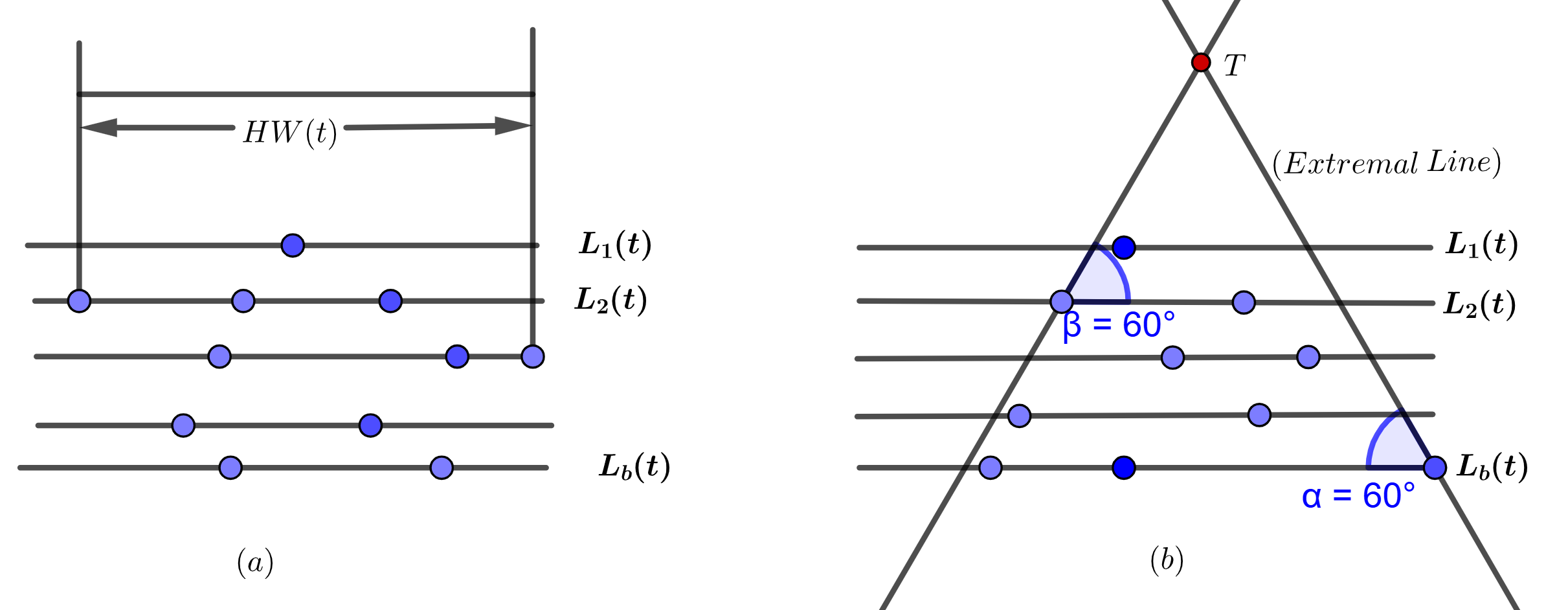}
    \caption{{\it (a). shows the horizontal width($HW(t)$) of the configuration $\mathbb{P}(t)$, and (b). shows how robots' positions, defining horizontal width and extremal lines, may not always be the same. }}
    \label{fig: figure 12}
\end{figure}
     \begin{itemize}
      \item \textbf{Horizontal Width $HW(t)$ :} For any configuration $\mathbb{P}(t)$ at time $t$, horizontal width is defined as the horizontal distance between two corner-most positions in $\mathbb{P}(t)$ as shown in the Figure~\ref{fig: figure 12}(a). 
      \item \textbf{Extremal Lines:} are the closest $60^\circ$, non-parallel lines such that all the robots in $\mathbb{P}(t_0)$ are deployed between both the lines, and positions of robots on these lines are defined as {\bf extremal positions/ points} as illustrated in the Figure~\ref{fig: figure 12}(b).

         \item \textbf{Computelevel():} For any robot $r_i$, this function takes $\mathcal{O}_i(t)$ as input and computes the set of horizontal lines $\mathcal{L}^i(t)$ and orders the set from north to south direction.
          \begin{figure}[h]
    \centering
    \includegraphics[width=0.4\textwidth]{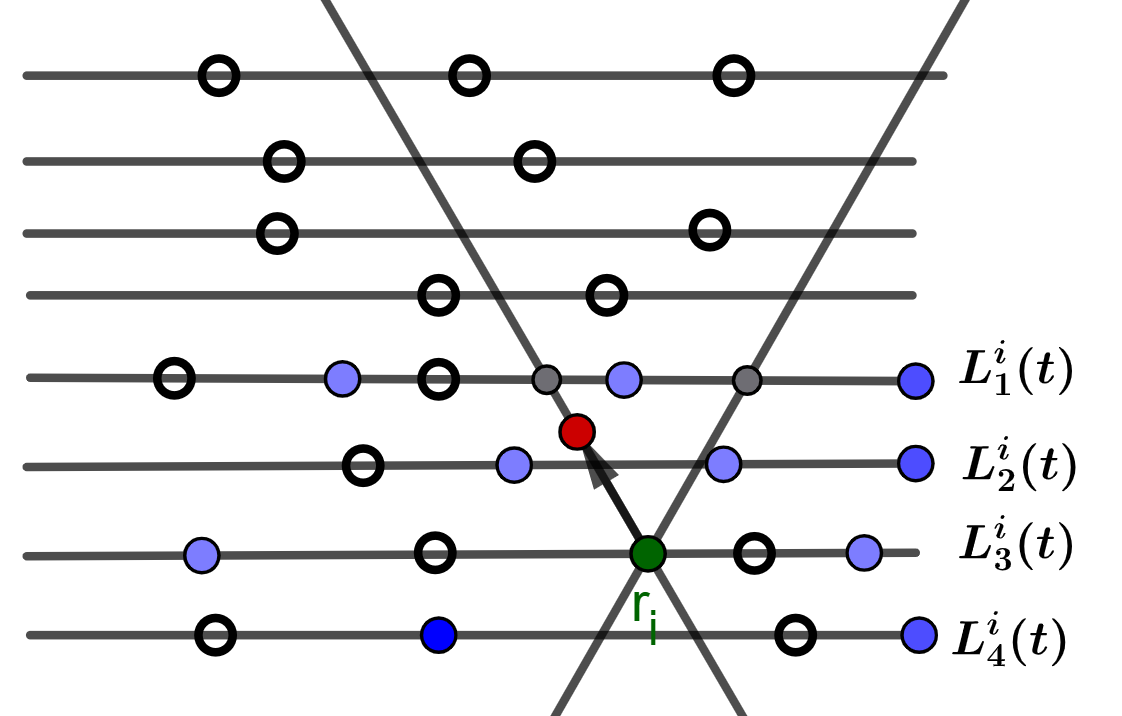}
    \caption{{\it Robot $r_i$ orders itself in $L_3^i(t)$ as it misses one robot in $L_1(t)$ and computes Go-Lines and moves accordingly towards the intersection point in the line $L_1^i(t) \simeq L_5(t)$. }}
    \label{fig: figure 13}
\end{figure}
         \item \textbf{ComputeGoLines():} If robot $r_i$ observes one or more robot positions above the horizontal line passing through $p_i(t)$, it invokes the function \texttt{ComputeGoLines()}, which computes two distinct Go-lines originating from $p_i(t)$ such that each makes a $60^\circ$ angle, computed clockwise from the horizontal lines passing through $p_i(t)$. These Go-lines ensure vertical progress even in highly adversarial situations and prevent $r_i$ from infinitely oscillating due to the adversarial defected view(as illustrated in Figure~\ref{fig: figure 13}).

          \item \textbf{Computecornerpoints():} For any robot $r_i$, this function takes $L^i_1(t)$ and returns corner positions on the north-most line $L^i_1(t)$.
        \item \textbf{ComputePerpendicular():} Using the function \texttt{ComputePerpendicular()}, robot $r_i$ computes perpendicular lines from the computed corner positions of the north-most horizontal line $L_1^i(t)$ to its current horizontal line. This allows $r_i$ to avoid unnecessary movement along the Go-line.

          \item \textbf{ComputeIntSec():} Using function \texttt{ComputeIntSec()}, $r_i$ computes intersection of perpendicular lines with $Go-lines$. 
          
          \item \textbf{Nearest():} Using this function, $r_i$ compare between two points. It takes both points as input and returns the nearest point from $p_i(t)$(current location).
     \end{itemize}
    \section{Algorithm~\ref{alg:movement}}

\subsection{High-Level Idea of the Algorithm}

In the adversarial $(N, K)$-defected view model, each robot works under strong limitations. During any activation, a robot may miss observing many other robots, its view may change arbitrarily between activations, and its movement can be interrupted by the adversary. In addition, robots are asynchronous, oblivious, and have no multiplicity detection capability. These constraints make coordination and gathering difficult.

The proposed algorithm overcomes these difficulties by using only simple geometric information that is always locally available and helps to make vertical despite the adversarial missing. The only common agreement among robots is the direction of the $Y$-axis. Using this shared direction, robots can consistently compare which observed positions are above or below them, even when their view is incomplete.

At a high level, each robot decides how to move based on the vertical ordering of the positions it can see. The algorithm follows a simple \emph{Go-Line strategy} that guarantees two kinds of progress. First, robots make \emph{vertical progress}: whenever possible, an active robot moves northward. This ensures that robots gradually align on higher horizontal levels. Importantly, a robot can make vertical progress even if it sees only one other robot.

Second, the algorithm ensures \emph{horizontal contraction}. The movement rules are designed so that the horizontal width $HW(t)$ of the swarm never increases and usually decreases over time. This contraction happens naturally as robots move while respecting the relative positions they observe.

The algorithm is robust to non-rigid motion. Even if a robot is stopped before reaching its destination, it always moves by at least a minimum distance $\delta > 0$ in the intended direction. As a result, every activation guarantees some positive vertical progress towards the computed north-most horizontal line. The proposed Go-Line strategy ensures no oscillation due to the adversarial defect in the vision.

By repeatedly applying these simple local rules, the configuration of the swarm shrinks over time. Once the configuration is sufficiently compact, gathering is completed in finite time. Thus, Algorithm~\ref{alg:movement} guarantees gathering for asynchronous, oblivious robots under the adversarial $(N, K)$-defected view model, using only agreement on the $Y$-axis and no additional coordination.

\subsection{Description of Algorithm~\ref{alg:movement}}
When a robot $r_i$ becomes active at time $t$, it observes a local configuration $\mathcal{O}_i(t)$, which may be incomplete due to the $(N, K)$-defected view model. Using the common agreement on the direction of the $Y$-axis, $r_i$ partitions the observed positions, including its own, into parallel horizontal lines and orders them from south to north. Based on its position within this ordering and on the geometric structure induced by the observed configuration, $r_i$ selects an appropriate movement rule. According to the different sub-cases defined by the number of observed robots and their distribution across the north-most horizontal line, $r_i$ computes a destination point using only local information and moves toward it, making a minimum progress of at least $\delta > 0$ despite possible interruptions due to non-rigid motion.

    \begin{itemize}
       \item \textbf{Case-A $\boldsymbol{|\mathcal{O}_i|=1}$:}
Robot $r_i$ observes no robot position other than its own. Hence, $r_i$ concludes that all robots are located at $p_i(t)$. Therefore, $p_i(t)$ is identified as the gathering point, and $r_i$ stops executing the algorithm.

            \item \textbf{Case-B $\boldsymbol{|\mathcal{O}_i|>1:}$} Suppose robot $r_i$ observes at least one robot position other than its own. Then, using the function \texttt{ComputeLevel()}, robot $r_i$ partitions the observed positions, including its own position, into horizontal levels and orders them from north to south. Based on its current level with respect to the observed positions, robot $r_i$ initiates its movement according to the corresponding rule of the algorithm.

          \begin{itemize}
        \item \textbf{Case-B(1): Robot $r_i$ is located on the topmost line.} Let robot $r_i$ be located at position $p_i(t)$. If robot $r_i$ does not observe any robot position strictly above the horizontal line passing through $p_i(t)$, then this line is identified as the topmost line and is denoted by $L_1^i(t)$.

        \begin{itemize}
           \item \textbf{Case-B(1.1):} If robot $r_i$ observes at least one robot position on a horizontal line strictly below the topmost line $L_1^i(t)$, then $r_i$ waits.

\item \textbf{Case-B(1.2):} Suppose robot $r_i$ is located at position $p_i(t)$ and observes all robot positions on the topmost line $L_1^i(t)$. In this case, robot $r_i$ further classifies its position into one of the following two categories.

         \begin{figure}[h]
    \centering
    \includegraphics[width=0.48\textwidth]{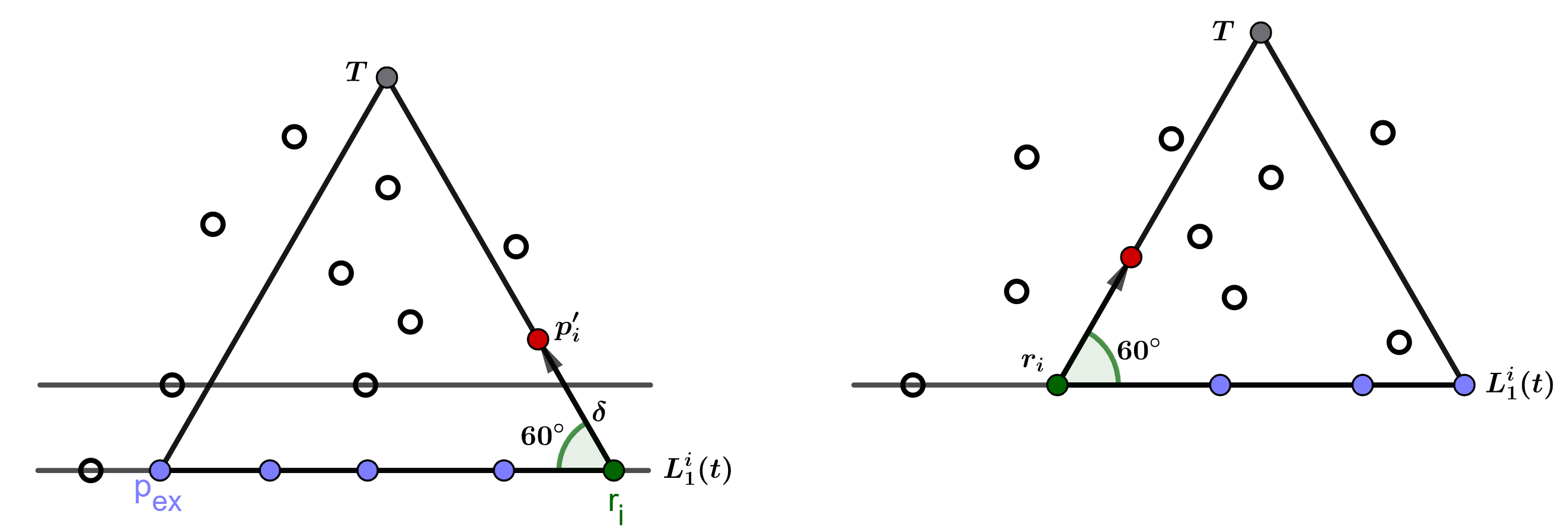}
    
    \caption{{\it In left, $r_i$ finds itself at the external position on the topmost line, and compute vertex of the equilateral triangle. On the right, $r_i$ misses corner one and finds itself at an external position, and computes the destination. }}
    \label{fig: figure 14}
\end{figure}
 \textbf{(i) Internal position:} If robot $r_i$ observes at least one robot position on both sides of $p_i(t)$ along the line $L_1^i(t)$, then $p_i(t)$ is identified as an \emph{internal position}. If $p_i(t)$ is an internal position, then $r_i$ remains stationary.

\textbf{(ii) External position:} If robot $r_i$ does not observe any robot position on at least one side of $p_i(t)$ along the line $L_1^i(t)$, then $p_i(t)$ is identified as an \emph{external position}. Let $p_j(t)$ be another external position observed by robot $r_i$. Using the function \texttt{ComputeTriangle()}, robot $r_i$ constructs an equilateral triangle $\triangle p_i T p_j$ such that the vertex $T$ lies strictly above the line $L_1^i(t)$. Robot $r_i$ sets $T$ as the destination point, and during the Move phase, robot $r_i$ moves toward the point $T$ along the line segment $\overline{p_i(t)T}$.\\
 \end{itemize}
        \item \textbf{Case-B(2): Robot $r_i$ is not located on the topmost line.} Suppose robot $r_i$, located at $p_i(t)$, observes a set of positions during the \emph{Look} phase. It stores these observations in $\mathcal{O}_i(t)$ and orders them using the auxiliary function \texttt{ComputeLevel()}. If there exists at least one observed position on a horizontal line strictly above the line $L_k^i(t)$ passing through $p_i(t)$, then $r_i$ identifies itself as being on the $k^{\text{th}}$ horizontal line. In this case, the active robot $r_i$ computes \textbf{Go-Lines($60^\circ$)}, consisting of two distinct rays originating from $p_i(t)$ and directed northward, each forming a $60^\circ$ clockwise angle with $L_k^i(t)$, as illustrated in Figure~\ref{fig: figure 16}.

    Subsequently, $r_i$ determines its path based on the geometric configuration of the visible robot positions on the north-most(topmost) horizontal line $L_1^i(t)$ with respect to the computed Go-Lines:

\end{itemize}

 \begin{itemize}
           
            \item \textbf{Case-B(2.1):} Suppose robot $r_i$ observes only one position, say $p_j(t)$ on the topmost line,i.e., $|L_1^i(t)| = 1$. 
If the robot position $p_j(t)$ lies on the vertical axis of robot $r_i$, then $r_i$ sets $p_j(t)$ as the destination and performs non-rigid motion towards $p_j(t)$ along its vertical axis.

Now suppose $p_j(t)$ is located on the vertical axis of $r_i$. Then, robot $r_i$ computes Go-lines using the auxiliary function \texttt{ComputeGoLines()}. If $p_j(t)\in L_1^i(t)$, such that $p_j(t)$ lies outside the Go-Lines, then robot $r_i$ computes the intersection of Go-Lines and the topmost horizontal line $L_1^i(t)$. Robot $r_i$ sets the nearest intersection point from $p_j(t)$ as the destination and moves towards it along the respective Go-Line.

 Now consider the case where a position $p_j(t)\in L_1^i(t)$ lies inside the Go-Lines. To avoid unnecessary movement and ensure that the horizontal span does not increase, robot $r_i$ computes a perpendicular from $p_j(t)$ to the horizontal line passing through $p_i(t)$ using \texttt{ComputePerpendicular()}. It then computes the intersection of this perpendicular with the corresponding Go-Line, sets this intersection point as its destination, and moves toward it along that Go-Line as illustrated in Figure~\ref{fig: figure 16}(c).

            \item \textbf{Case-B(2.2):} Suppose robot $r_i$ observes more than one robot position on the topmost horizontal line $L_1^i(t)$,i.e., $|L_1^i(t)|>1$, then according to their deployment with respect to the Go-Lines originating from $p_i(t)$.
             \begin{figure}[h]
    \centering
    \includegraphics[width=0.48\textwidth]{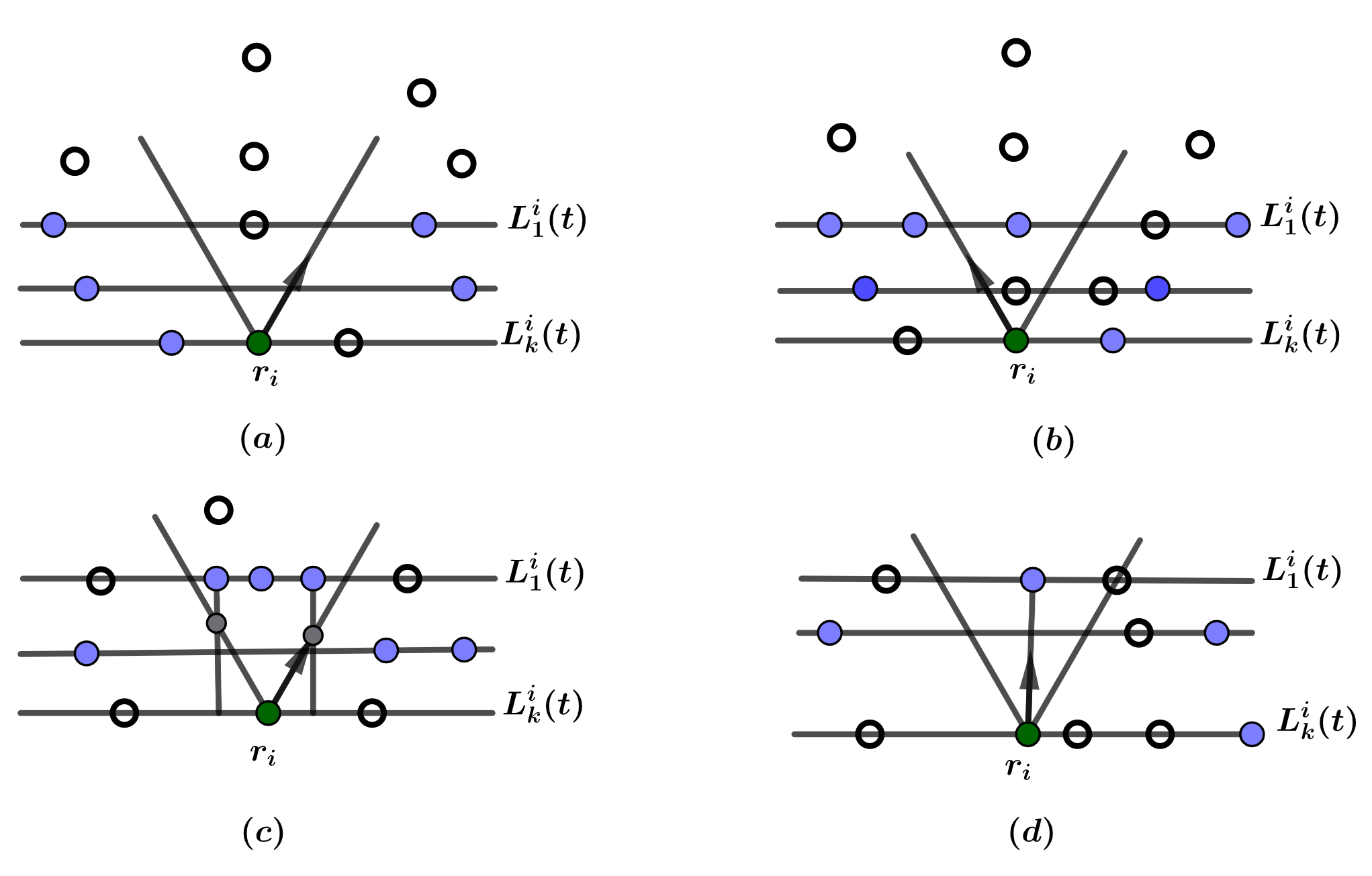}
    \caption{{\it Showing movement of active robot along the Go-Line accordingly.}}
    \label{fig: figure 16}
\end{figure}
            \begin{itemize}
                \item \textbf{Case-B(2.2.a):} If both Go-Lines intersect the topmost horizontal line $L_1^i(t)$ at points lying between the visible positions on that line, then $r_i$ computes the two intersection points between the Go-Lines and $L_1^i(t)$. Let these intersection points be denoted by $d_i$ and $d_i'$. The robot $r_i$ then selects one of these points as a destination arbitrarily and moves at least $\delta>0$ toward the destination along the corresponding Go-Line, as illustrated in Figure~\ref{fig: figure 16}(a).

               \item \textbf{Case-B(2.2.b):}If exactly one Go-Line intersects the topmost horizontal line $L_1^i(t)$ at a point lying between the visible positions on that line, while the other Go-Line intersects $L_1^i(t)$ outside the range of the visible positions, then $r_i$ moves toward the intersection point, denoted by $d_i$, along the Go-Line whose intersection lies between the visible positions, as illustrated in Figure~\ref{fig: figure 16}(b).

               \item \textbf{Case-B(2.2.c):}If both Go-Lines originating from $r_i$ intersect the topmost horizontal line $L_1^i(t)$ at points lying outside the range of the visible positions, then $r_i$ constructs perpendicular lines from the corner positions of $L_1^i(t)$ onto the line passing through $p_i(t)$ using the function \texttt{ComputePerpendicular()}. It then computes the intersection points of these perpendiculars with the corresponding Go-Lines, denoted by $d_i$ and $d_i'$. Finally, $r_i$ selects the closer of the two points using \texttt{Nearest($d_i, d_i'$)} and moves toward it along the associated Go-Line, as illustrated in Figure~\ref{fig: figure 16}(c) and (d).

        \end{itemize}
        \end{itemize}
        \end{itemize}
    Every active robot follows the proposed Algorithm repeatedly until it gathers at the same location. If robots stop watching any robot position other than itself, then the Algorithm terminates.
    
\begin{algorithm}[!ht]
\caption{Movement Algorithm for Robot $r_i$}
\label{alg:movement}
\begin{algorithmic}[1]
\footnotesize
\Require Observation set $\mathcal{O}_i(t)$, lines $L^i_1(t)$, $L^i_k(t)$
\Ensure Destination $d_i$ for robot $r_i$

\If{$|\mathcal{O}_i(t)| = 1$}
    \State $d_i \gets p_i(t)$ \Comment{No visible robots, wait}

\Else
    \State Assume $|\mathcal{O}_i(t)| > 1$

    \If{$r_i \in L^i_1(t)$}

        \If{$r_i$ observes positions below $L^i_1(t)$}
            \State $d_i \gets p_i(t)$ \Comment{wait}

        \ElsIf{$r_i$ sees all positions on $L^i_1(t)$}

            \If{$r_i$ is at an internal position}
                \State $d_i \gets p_i(t)$ \Comment{wait}
            \Else
                \State Let $p_{ex}$ be another external visible position on $L^i_1(t)$
                \State $d_i \gets \texttt{ComputeVertex}(p_i(t), p_{ex})$
            \EndIf

        \EndIf

    \ElsIf{$r_i \in L^i_k(t)$ \textbf{and} $r_i$ observes positions above $L^i_k(t)$}

        \State $\text{CornerPts} \gets \texttt{ComputeCorner}(L^i_1(t))$
        \State $\text{GoLines} \gets \texttt{ComputeGoLines}(L^i_1(t))$

        \If{$|L^i_1(t)| = 1$}
            \State Let $p_j \in L^i_1(t)$
            \If{$p_j$ lies on the vertical axis of $r_i$}
                \State $d_i \gets p_j$
            \Else
                \State Compute perpendicular from $p_j$ to $L^i_k(t)$
                \State $d_i \gets \texttt{ComputeIntSec}(\text{perpendicular}, \text{GoLine})$
            \EndIf

        \Else

            \If{all points of $L^i_1(t)$ lie outside Go-Lines}
                \State $d_i \gets (Go\text{-}Line) \cap L_1^i(t)$
                \State $r_i$ moves towards $d_i$ along one of the Go-Lines

            \ElsIf{exactly one Go-Line intersects $L_1^i(t)$ between visible positions}
                \State $d_i \gets (Go\text{-}Line) \cap L_1^i(t)$
                \State $r_i$ moves towards $d_i$ along that Go-Line

            \ElsIf{both Go-Lines lie between visible positions}
                \State $(d_i, d_i') \gets (Go\text{-}Line) \cap L_1^i(t)$
                \State $r_i$ moves towards one of $d_i$ or $d_i'$ along the respective Go-Line

            \Else
                \State Compute perpendiculars from corner points of $L^i_1(t)$
                \State $(d_i,d_i') \gets \texttt{ComputeIntSec}(\text{perpendicular}, \text{GoLine})$
                \State $r_i$ moves towards $\texttt{Nearest}(d_i,d_i')$ along the respective Go-Line
            \EndIf

        \EndIf
    \EndIf

\EndIf

\State \Return $d_i$

\end{algorithmic}
\end{algorithm}

\subsection{Correctness}

In this section, we prove the correctness of Algorithm~\ref{alg:movement} for the gathering problem under the adversarial $(N,K)$-defected view model, where $1 \leq K \leq N-2$. 
Despite the severely constrained observation setting—in which each active robot may fail to observe up to $N-2$ other robots and the set of observed positions may vary arbitrarily between successive \emph{Look} phases, we show that the algorithm guarantees that all robots eventually gather at a single common location.

The remainder of this section formalizes these observations and presents a sequence of lemmas that collectively establish that Algorithm~\ref{alg:movement} correctly solves the gathering problem under the adversarial $(N, K)$-defected view model for any $N \geq 3$, relying solely on agreement along the vertical ($Y$) axis and without requiring any additional robot capabilities.

\begin{lemma}{\it Regardless of adversarial missing, Go-Lines ensures the guaranteed vertical progress of active robots along the upward direction. }
\begin{proof} Every active robot located below $L_1()$ always moves either along its vertical axis or along the computed Go-Lines. First, let $r_i\in L_k^i()$ be any active robot and $|L_1^i()|=1$, say $p_j$, then if $p_j$ has the same $X-$coordinate as $p_i$, then $r_i$ moves at least $\delta$ distance towards $p_j$ along its vertical axis as in Figure~\ref{fig: figure 16}(d), otherwise $r_i$ always move along the computed Go-Line. Every active robot $r_i$ that sees locations above its horizontal line always computes Go-Lines($60^\circ$). According to the proposed Algorithm, if the robot $r_i\in L_k^i()$ computes its next destination at the visible topmost line $L_1^i()$, then it moves at least $\delta$ along the $60^\circ$ line, so in each move $y-$coordinate of $ r_i$'s location increases by at least $\delta\frac{\sqrt{3}}{2}$ units. If all the positions on $L_1^i()$ located between the intersection points of Go-Lines and $L_1^i()$ as in Figure~\ref{fig: figure 16}(c), then $r_i$ computes perpendicular lines using \texttt{Computeperpendicular()} from corner positions located on $L_1^i()$ to $L_k^i()$, and where these perpendicular lines intersects Go-Lines, $r_i$ chooses the nearest intersection point on one of the Go-Line and moves towards it. When the robot finds itself at an external position on the topmost line and does not see any position below, then it computes an equilateral triangle with the other observed external position and moves towards the upward direction along the non-horizontal side of the triangle, and progresses at least $\delta\frac{\sqrt{3}}{2}$ unit along the upward direction and again robot below chase the topmost positions. The process continues until all the robots gather at the same location.  \\
Go-lines are robust to adversarial hiding, as moving along the Go-lines does not depend on how many robots any active robot looks or misses at the topmost line; the proposed strategy is concerned only with their positions with respect to the computed Go-Lines. By the adversarial hiding, the robot may change its Go-Line in each round, but it makes at least $\delta\frac{\sqrt{3}}{2}$ units progress towards the visible topmost line to $r_i$.
\end{proof}
\end{lemma}

\begin{lemma}{\it Horizontal width $HW()$ of robot positions decreases monotonically,  $HW(t+1)\leq HW(t)$. }
\begin{proof} Every active robot moves towards the upward direction along the $60^\circ$ Go-Lines according to the Algorithm~\ref{alg:movement}. First, consider if the horizontal width of the configuration is determined by the extremal robots, then if the bottom-most extremal robot moves $\delta$ distance along the extremal line, then it minimizes the total width by at least $\frac{\delta}{2}$ units. The Algorithm is designed in such a way that non-extremal robots never move outside the extremal line. Suppose $r_i\in L_k()$ is a non-extremal robot and Go-Lines from $p_i()$ intersects $L_1^i()$ outside the span of $L_1^i()$, then  $r_i$  computes perpendicular lines computed from corner-most positions of $L_1^i()$ on $L_k^i()$  and computes intersection of perpendicular lines and Go-Lines and moves to one of the intersection point as shown in the figure. If the corner-most position of $L_1^i()$is located on the Go-line and it is also an extremal point, then only $r_i$ can reach the extremal line but can not cross it in any situation. If the horizontal width of the configuration is not determined by the extremal robots, then they behave like a non-extremal robot. The process continues until all the robots reach the common destination. The extremal robots move up, and the non-extremal robots chase them. In each round, the horizontal width of the robot positions decreases monotonically, i.e., $HW(t)\leq HW(t+1)$. \end{proof}
\end{lemma}

\begin{lemma}{\it If the distance between extremal points on $L_1(t_0)$ is $d$, then all the robots definitely gather at the intersection  $T$ of both extremal lines, located at $Y_{new}=Y_{max}(0)+\frac{\sqrt{3}}{2}d$.}
\begin{proof} Let at $L_1(t_0)$, distance between extremal points is $d$. Both the extremal lines are non-parallel, and each makes $60^\circ$ angle with $L_1(t_0)$, so they must intersect each other at some point $T$ located $\frac{\sqrt{3}}{2}d$ units above the line $L_1(t_0)$. From Lemmas 7 and 8, we can see the guaranteed progress of robots along the upward direction without Going outside the extremal lines. Consider the action of extremal robots. If any extremal robot computes the equilateral triangle with a non-extremal position, then it always computes the vertex at the extremal line along which it is moving. If it computes the triangle vertex with extremal robots, then it must be the intersection $T$ of both the extremal lines. The Algorithm guarantees that no robot moves across $T$. If any one of the extremal robots reaches $T$, then it will be at the topmost position and will wait for the other robots. Robots on the same and another extremal line must reach $T$ in finite time and will wait for non-extremal robots. The non-extremal robots will be on the same vertical axis after a finite number of steps because of the perpendicular-line rule imposed by the Algorithm, and reach $T$ in finite time. So the gathering can be achieved at $T$ located $\frac{\sqrt{3}}{2}d$ units above the $ Y_{max}(0)$ as shown in Figure~\ref{fig: figure 15}.
  \begin{figure}[h]
    \centering
    \includegraphics[width=0.45\textwidth]{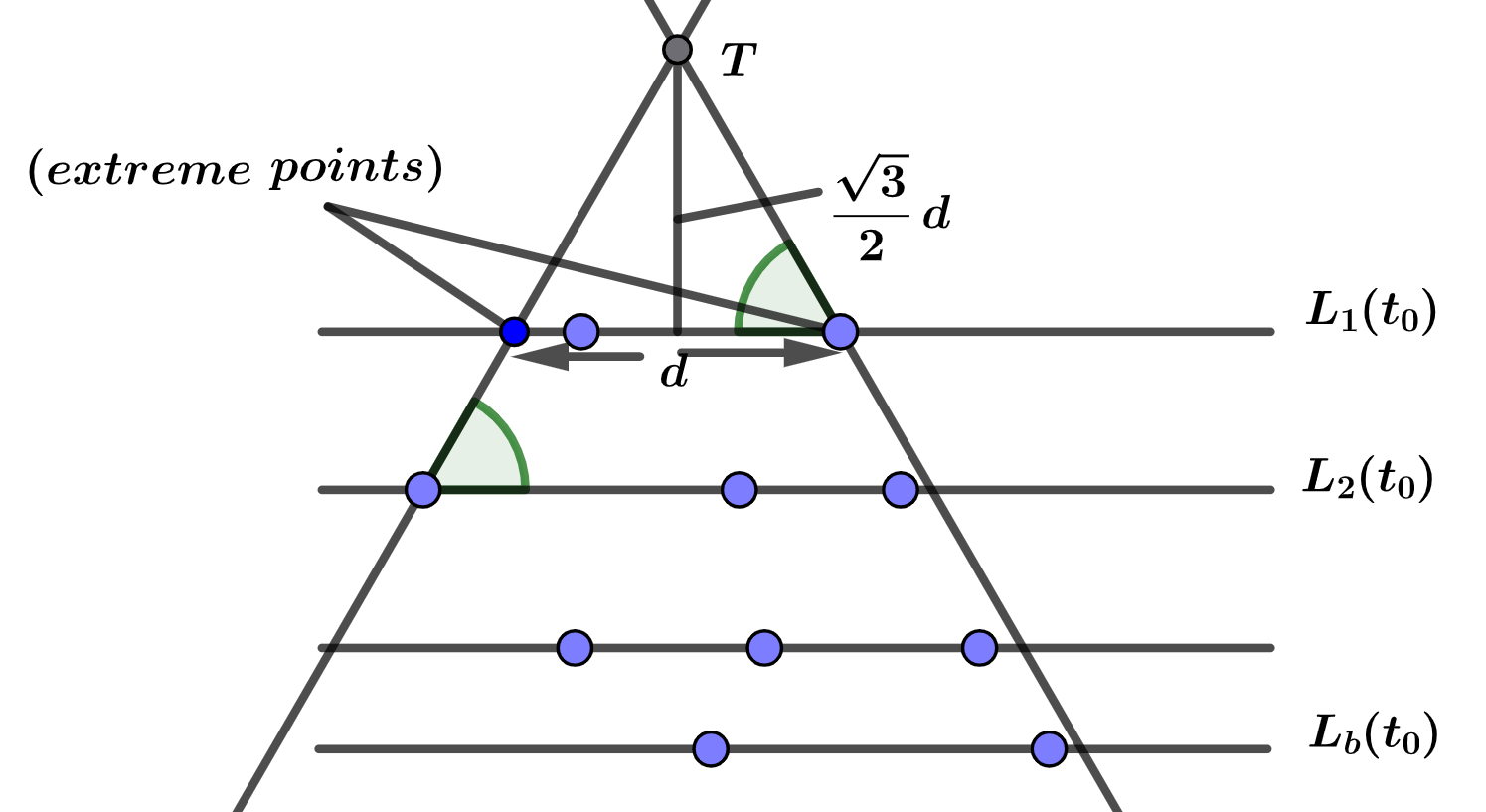}
    \caption{{\it As the distance between extremal lines on the line $L_1(t_0)$ is $d$, then they must intersect each other at $\frac{\sqrt{3}}{2}d$ units above $L_1(t_0)$.}}
    \label{fig: figure 15}
\end{figure}
\end{proof}
\end{lemma}
From Lemmas 2.1, 2.2, and 2.3, it follows that-\\

\begin{theorem} {\it In the adversarial $(N, K)$ defected view model $(N\geq3)$, Algorithm~\ref{alg:movement} solves gathering for asynchronous robots without multiplicity detection, provided they have agreement on the vertical axis of the local coordinate system.}
\end{theorem}

\section{Simulation Results and Experimental Validation}
\label{sec:simulation}

To empirically validate the correctness and convergence properties of the proposed gathering algorithm under defected visibility models, we implemented a discrete-time simulator in Python. These configurations are uniformly randomly generated for simulation purposes. Robots operate under the standard Look–Compute–Move cycles. At each round, the defected view is generated according to the corresponding adversarial visibility model. Unless otherwise stated, robots are placed in the Euclidean plane with arbitrary initial positions. The execution terminates when all robots occupy the same position. To quantify gathering, we measure the spatial span of the configuration over time. In particular, for a configuration at round $t$, we define:


\subsection{$(4,2)$ Adversarial Defected View Model}
\label{subsec:42model}

We first evaluate the algorithm under the $(4,2)$ adversarial defected view model.

\noindent \textbf{Collinear Initial Configuration:}
Figure~\ref{fig:areavssteps}(a) shows that under the FSYNC model ($N=4$, $K=2$), both the initial configuration and the observed robot subsets are generated pseudorandomly (with a fixed seed for reproducibility). Since all robots are initially collinear, no vertical alignment is required, and the vertical span remains zero throughout. The horizontal span decreases monotonically and almost linearly across synchronous rounds, indicating uniform global contraction. Even under non-rigid motion and bounded visibility, each round guarantees consistent progress, demonstrating stable and predictable convergence under full synchrony.

\noindent \textbf{Non-Collinear Initial Configuration:}
Figure~\ref{fig:areavssteps}(b) illustrates the evolution of the convex hull area for $N=4$ and $K=2$ under the FSYNC model, where the robots are initially placed pseudorandomly in a non-collinear configuration. The observed subsets for each robot are also generated pseudorandomly (with a fixed seed for reproducibility). The area of the convex hull decreases monotonically across synchronous rounds, indicating uniform global contraction of the configuration(an example of area contraction is shown in Figure~\ref{fig:noncollinearrounds}(a), (b), (c)). The smooth and progressively flattening decay suggests consistent collective movement toward gathering, with faster reduction in early rounds due to larger initial dispersion and slower contraction near gathering point as robots approach proximity.
\begin{figure*}[htb]
\centering

\begin{minipage}[t]{0.48\textwidth}
    \centering
    \includegraphics[width=0.95\linewidth]{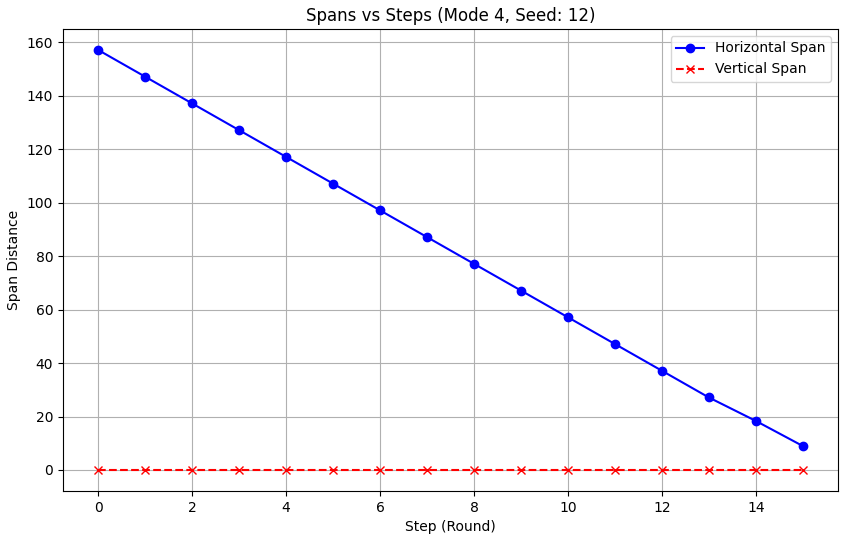}
    
    \vspace{2mm}
    (a)
\end{minipage}
\hfill
\begin{minipage}[t]{0.48\textwidth}
    \centering
    \includegraphics[width=0.95\linewidth]{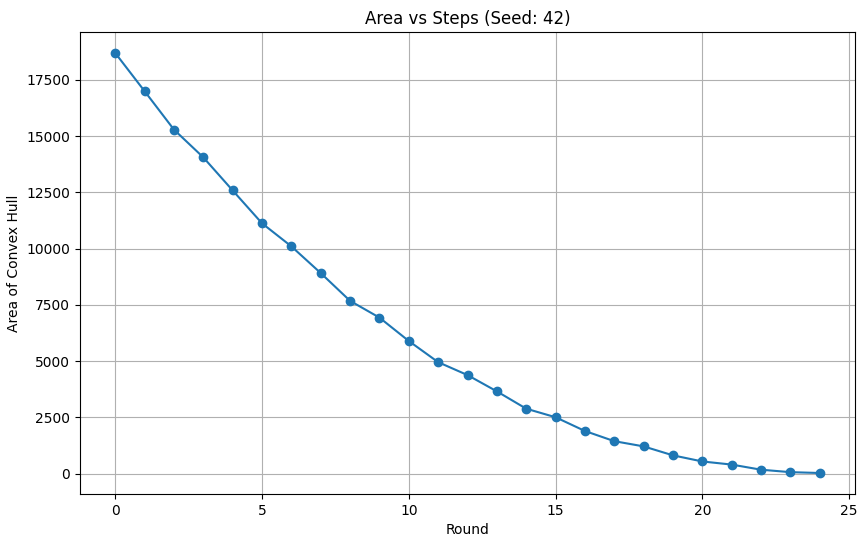}
    
    \vspace{2mm}
    (b)
\end{minipage}

\caption{(a) Illustrates that when robots are collinear, then horizontal width decreases monotonically and vertical span remains zero throughout the execution. (b) Shows that the area formed by the convex hull of non-collinear positions decreases as the number of rounds increases.}
\label{fig:areavssteps}
\end{figure*}
\subsection{$(N,K)$ Defected View Model}
\label{subsec:NKmodel}
The theoretical correctness of Algorithm~\ref{alg:movement} is established through rigorous mathematical analysis covering all admissible asynchronous executions and adversarial visibility defects. The simulation study, therefore, does not aim to re-validate correctness, but rather to provide intuitive insight into the algorithm’s geometric behavior, local decision-making process, and empirical scaling properties.

The Python implementation is organized around two principal classes: \texttt{Robot} and \texttt{GatheringSystem}. Each robot maintains its Cartesian position, trajectory history, visible set, and computed destination. The motion strictly follows the non-rigid model: when activated, a robot moves toward its destination by at most a fixed distance $\delta$, ensuring guaranteed minimum progress while allowing partial movement, consistent with the theoretical assumptions.

The global execution is handled by the \texttt{GatheringSystem}. At each asynchronous step, a nonempty subset of robots is selected using a seeded \emph{pseudorandom} scheduler to approximate asynchrony in a reproducible manner. For every activated robot, a visibility bound $K \in [1, N-2]$ is chosen pseudorandomly, and at most $K$ robots (including itself) are exposed in its view. The geometric rules of the algorithm are then applied exactly as specified: visible robots are partitioned into horizontal layers; robots on the topmost layer either wait, remain stationary if internally positioned, or compute an equilateral-triangle apex if extremal; robots below the topmost layer move along $60^\circ$ Go-lines toward appropriate intersections (an example of such movement is shown in the Figure~\ref{fig:movementNK}). Thus, the simulation directly encodes all geometric cases under adversarially varying visibility.

Convergence is evaluated using two global measures: horizontal width $(\max x - \min x)$ and vertical span $(\max y - \min y)$. Convergence is declared once these quantities fall below a fixed threshold (here $10^{-1}$). For each swarm size $N \in \{5,7,\dots,49\}$, multiple independent runs with different pseudorandom seeds are executed, and the average number of asynchronous steps required for convergence is plotted.

The resulting graphs (as illustrated in Figure~\ref{fig:Knotfixed}) show a clear and consistent pattern. The horizontal convergence curve is obtained by averaging over $5$ independent pseudorandom runs for each swarm size, whereas the vertical convergence curve is averaged over $17$ independent runs. This averaging mitigates fluctuations caused by the pseudorandom asynchronous scheduler and adversarial visibility selection, thereby revealing the intrinsic geometric behavior of the algorithm rather than run-specific variations.

Horizontal convergence remains tightly bounded (approximately $500$--$650$ steps) and exhibits only mild growth with increasing $N$. This behavior reflects the strength of the extremal-line invariant: only extremal robots actively reduce the global width, while internal robots do not contribute to its expansion. As a result, horizontal contraction progresses in a controlled and nearly size-insensitive manner. The relatively small dispersion across runs further indicates that horizontal gathering is structurally stable and weakly affected by adversarial visibility variations.

In contrast, vertical convergence shows a more pronounced dependence on swarm size, reaching several thousand steps for larger values of $N$. This scaling behavior is consistent with the layered geometric contraction mechanism, where vertical reduction occurs through successive elimination of upper horizontal layers. Larger swarms generally introduce greater initial vertical dispersion and more intermediate layers, which increases the number of contraction phases required. Nevertheless, the averaged curve remains smooth and monotonic in expectation, with no oscillatory or divergent behavior observed. This confirms that the increased convergence time is a consequence of geometric scaling rather than instability in the asynchronous execution.

It is important to emphasize that the simulation does not attempt to exhaustively enumerate all asynchronous executions; such guarantees are already provided by the theoretical proofs. Instead, the simulation offers a visual and quantitative illustration of how local geometric rules collectively produce global contraction. The consistent finite-time convergence across all tested swarm sizes provides intuitive confirmation of the algorithm’s robustness and scalability under adversarial $(N,K)$-defected visibility.
\begin{figure*}[htb]
\centering

\begin{minipage}{0.32\textwidth}
    \centering
    \includegraphics[width=\linewidth]{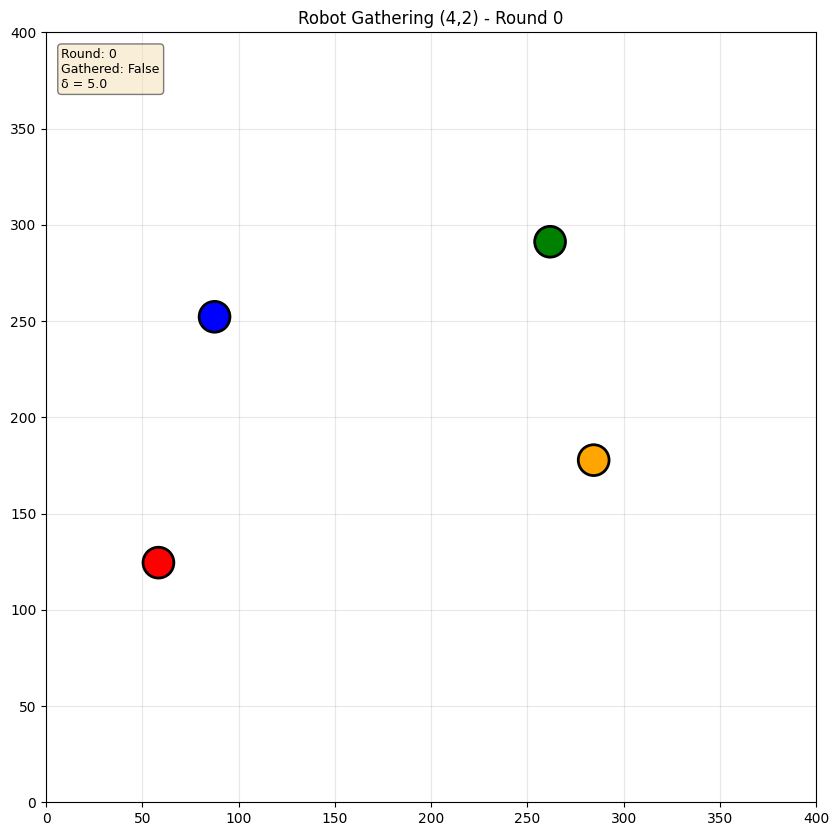}
    
    \vspace{2mm}
    (a)
\end{minipage}
\hfill
\begin{minipage}{0.32\textwidth}
    \centering
    \includegraphics[width=\linewidth]{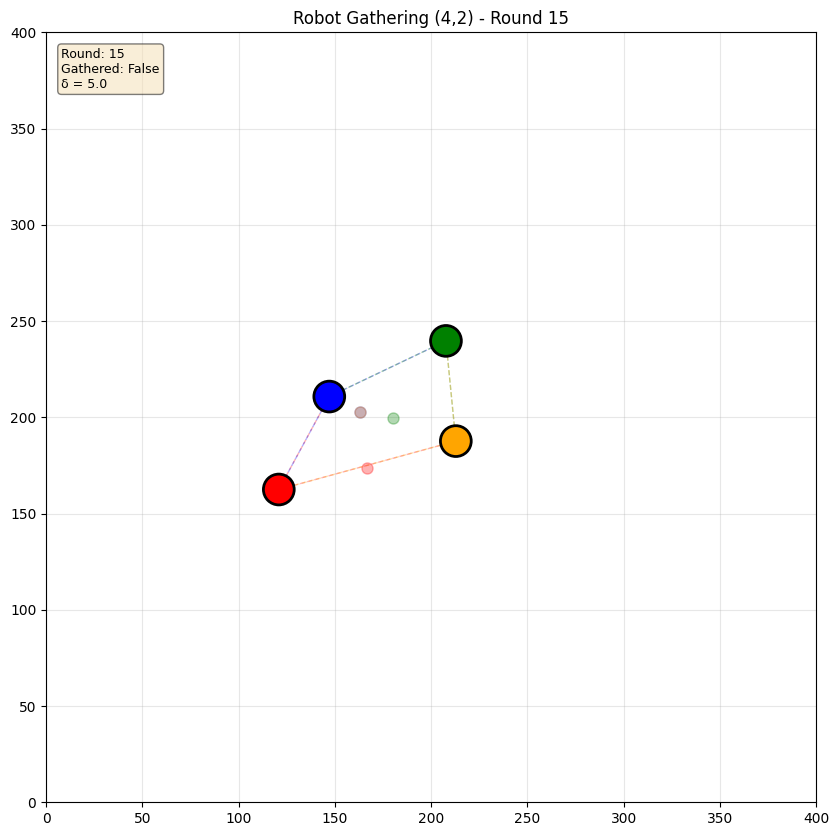}
    
    \vspace{2mm}
    (b)
\end{minipage}
\hfill
\begin{minipage}{0.32\textwidth}
    \centering
    \includegraphics[width=\linewidth]{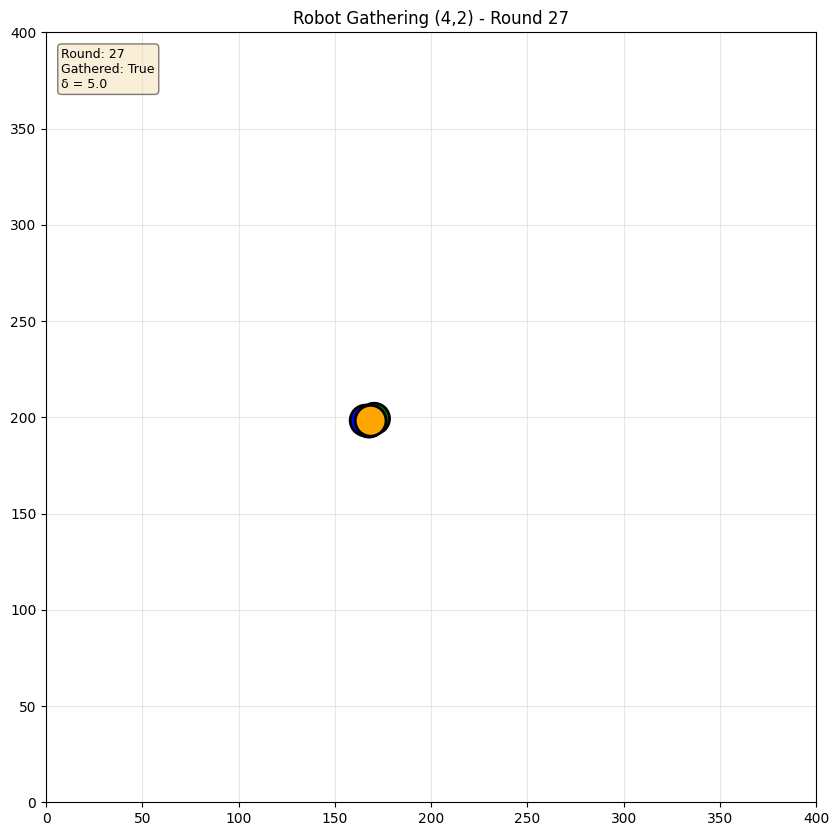}
    
    \vspace{2mm}
    (c)
\end{minipage}

\caption{Evolution of a non-collinear configuration under the adversarial $(4,2)$ defected view model: (a) Round~0, (b) Round~15, and (c) Round~27.}
\label{fig:noncollinearrounds}
\end{figure*}

\begin{figure*}[htb]
\centering

\begin{minipage}{0.45\textwidth}
    \centering
    \includegraphics[scale=0.35]{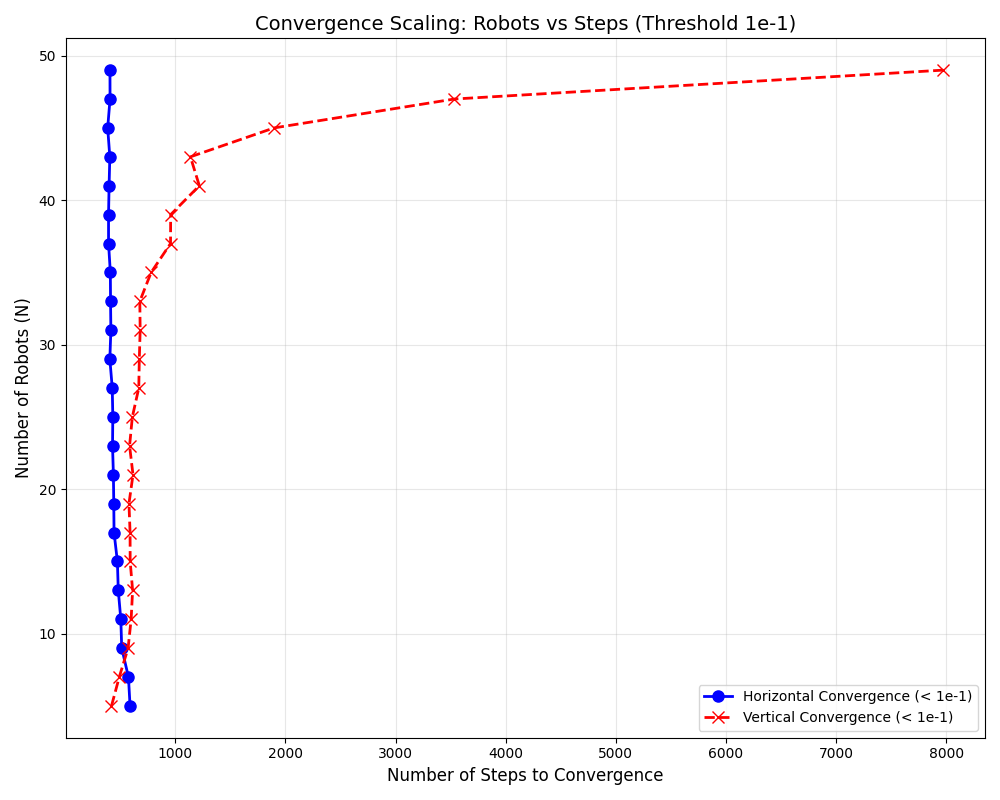}
    
    \vspace{2mm}
    (a) 5-run
\end{minipage}
\hfill
\begin{minipage}{0.45\textwidth}
    \centering
    \includegraphics[scale=0.35]{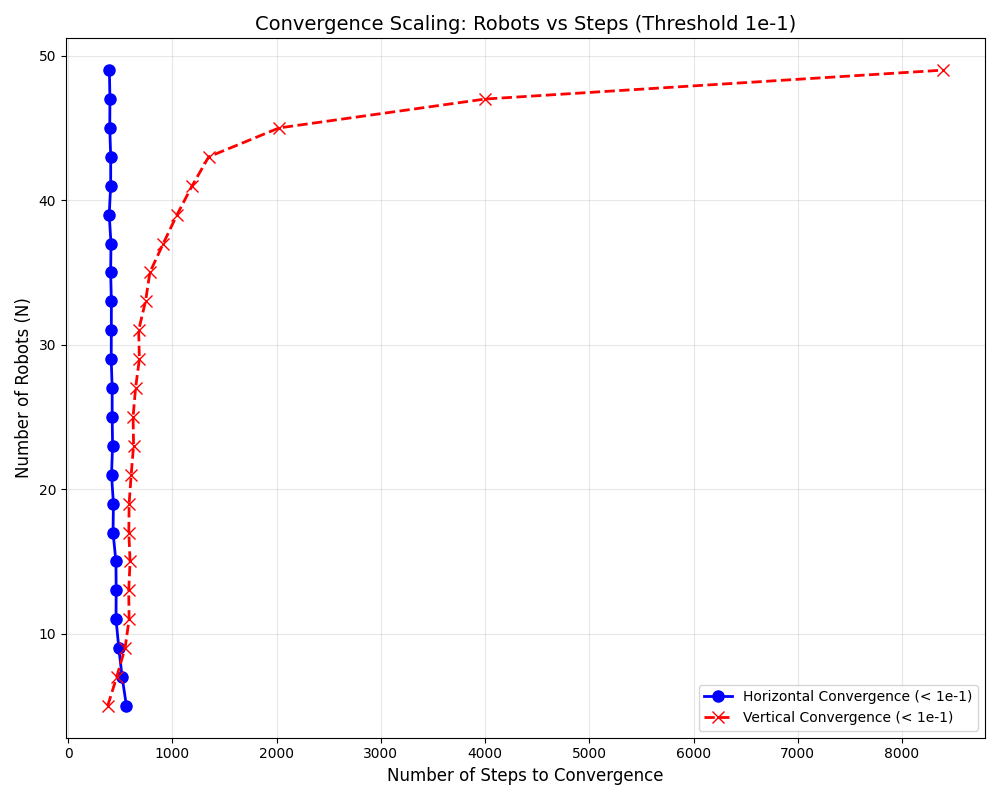}
    
    \vspace{2mm}
    (b) 17-run
\end{minipage}

\caption{The 17-run average yields a smoother and more consistent decay of vertical span under adversarial defect sequences, and both the graph shows stable horizontal convergence.}
\label{fig:Knotfixed}
\end{figure*}

\begin{figure*}[!htbp]
\centering

\begin{minipage}{0.25\textwidth}
    \centering
    \includegraphics[scale=0.34]{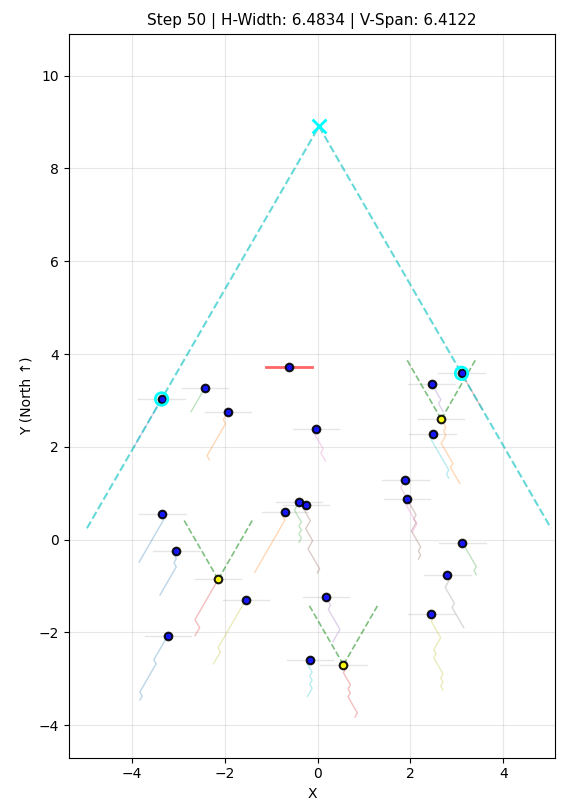}
    
    \vspace{2mm}
    (a)
\end{minipage}
\hfill
\begin{minipage}{0.3\textwidth}
    \centering
    \includegraphics[scale=0.34]{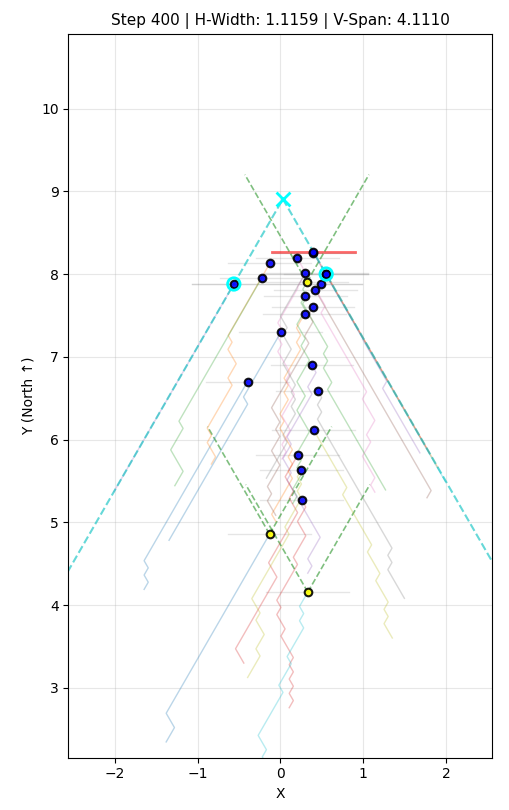}
    
    \vspace{2mm}
    (b)
\end{minipage}
\hfill
\begin{minipage}{0.32\textwidth}
    \centering
    \includegraphics[scale=0.34]{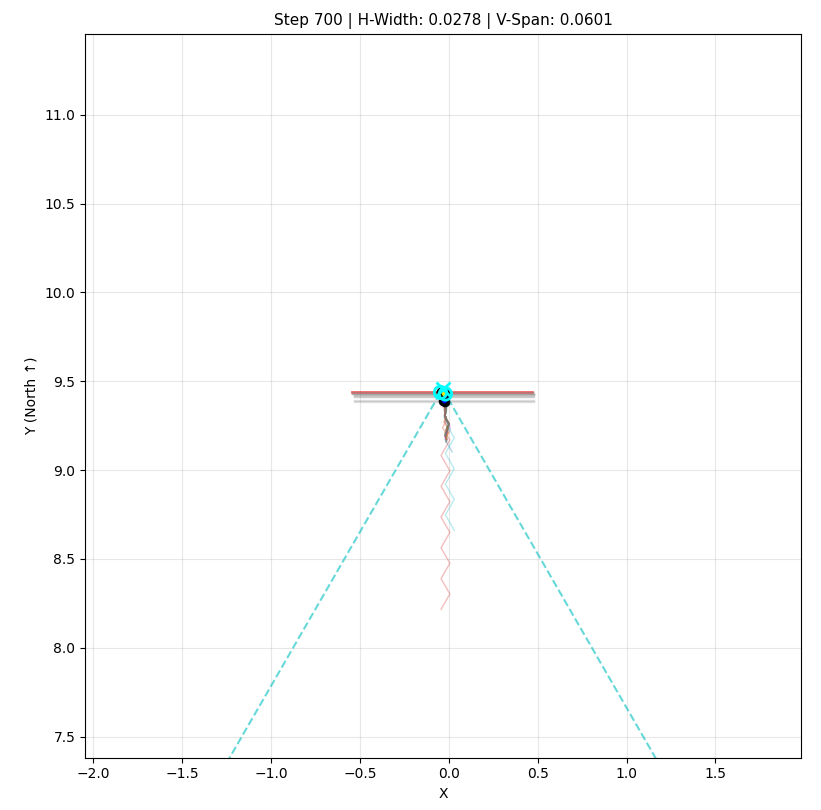}
     \vspace{2mm}
    (c)
\end{minipage}

\caption{Blue lines denote the extremal lines whose intersection defines the gathering point. Robots move toward the computed destination, resulting in contraction of the horizontal span. The experiment corresponds to $N=25$ with dynamically varying observation parameter $K \in [1,23]$.}
\label{fig:movementNK}
\end{figure*}

  \section{Conclusion}

In this paper, we have studied the gathering problem for autonomous robot swarms operating under the recently introduced \emph{defected view model}, where robots may observe only a limited and adversarially selected subset of other robots. This model captures severe sensing restrictions and extends classical visibility assumptions commonly adopted in the literature.

We first addressed an open problem concerning the gathering of four fully synchronous robots under the adversarial $(4,2)$ defected view model cited in ~\cite{kim2023gathering}. For this setting, we presented a distributed Algorithm that guarantees gathering in finite time. The proposed algorithm operates without assuming any additional capabilities such as multiplicity detection, memory, or agreement on the local coordinate system, thereby establishing that gathering is achievable under minimal assumptions even in the presence of adversarial observation.

Then, we have investigated the gathering problem for \textbf{asynchronous} robots under the adversarial $(N,K)$ defected view model for \emph{all} values $1 \le K \le N-2$. In particular, we presented and analyzed a finite-time distributed Algorithm that works even in the extreme situations where a robot may miss $N-2$ out of the other $N-1$ robots and can observe only a single other robot at any activation. The proposed Algorithm possesses a finite gathering assuming that robots share a common understanding of only \textbf{One-axis} of their local coordinate systems. This result substantially strengthens existing solvability bounds and demonstrates the surprising robustness of gathering under asynchronous execution and highly adversarial sensing conditions. Unlike most existing works on the defected view model, our Algorithms operate under the \textbf{non-rigid movement model}. We have strengthened the robustness and practical relevance of our results.

Several challenging directions remain for future research. One natural extension is to relax the assumption of axis agreement and investigate the feasibility of gathering under the defected view model without any agreement on the local coordinate systems. Another important direction is to study similar problems in higher-dimensional Euclidean spaces, where geometric obstructions and visibility constraints become more complex. Additionally, exploring gathering under weaker robot capabilities or alternative obstruction and visibility assumptions may further deepen the understanding of the fundamental limits of coordination in defected observation models.


\bibliographystyle{plain} 
\bibliography{bibliography.bib}

@String{Computing = "Computing" }

@String{Computer = "{IEEE} Computer" }

@String{Springer = "Springer-Verlag" }

@article{di2025gathering,
  title={Gathering on a circle with limited visibility by anonymous oblivious robots},
  author={Di Luna, Giuseppe Antonio and Uehara, Ryuhei and Viglietta, Giovanni and Yamauchi, Yukiko},
  journal={Theoretical Computer Science},
  volume={1025},
  pages={114974},
  year={2025},
  publisher={Elsevier}
}

@article{Bhagat201650,
  author    = {Subhash Bhagat and Sruti Gan Chaudhuri and Krishnendu Mukhopadhyaya.},
  title     = {Fault-tolerant gathering of asynchronous oblivious mobile robots under one-axis agreement},
  journal   = {J. Discrete Algorithms},
  volume    = {36},
  pages     = {pages 50--62},
  year      = {2016},
 issn = "1570-8667",
  url       = {http://www.sciencedirect.com/science/article/pii/S1570866715001070},
  doi       = {http://dx.doi.org/10.1016/j.jda.2015.10.005},
  bibsource = {dblp computer science bibliography, http://dblp.org}
}

@inproceedings{BhagatM17,
  author    = {S. Bhagat and
               K. Mukhopadhyaya},
  title     = {Fault-tolerant Gathering of Semi-synchronous Robots},
  booktitle = {Proc. the 18th International Conference on Distributed Computing
               and Networking (ICDCN-2017)},
  pages     = {6},
  year      = {2017}
}

@inproceedings{BhagatCM18,
  author    = {S. Bhagat and
               S. Gan Chaudhuri and
               K. Mukhopadhyaya},
  title     = {Gathering of Opaque Robots in 3D Space},
  booktitle = {Proc. the 19th International Conference on Distributed Computing and Networking, (ICDCN-2018)},
  pages     = {2:1--2:10},
  year      = {2018}
}

@article{agmon2006fault,
  title={Fault-tolerant gathering algorithms for autonomous mobile robots},
  author={Noa Agmon and David Peleg},
  journal={SIAM Journal on Computing},
  volume={36},
  number={1},
  pages={ pages 56--82},
  year={2006},
  publisher={SIAM}
}

@INPROCEEDINGS{Quentin2014,
author={Quentin Bramas and Sébastien Tixeuil.},
booktitle={Proc. International Colloquium on Structural Information and Communication Complexity (SIROCCO)}, 
title={Wait-free gathering without chirality},
year={2014},
month={},
volume={},
number={},
pages={313-327},
doi={},
ISSN={},}

@inproceedings{bhagat2017gathering,
  title={Gathering Asynchronous Robots in the Presence of Obstacles},
  author={Bhagat, Subhash and Mukhopadhyaya, Krishnendu},
  booktitle={11th International Conference and Workshop on Algorithms and Computation (WALCOM)},
  pages={279--291},
  year={2017},
}

@article{CiceroneSN18,
  author    = {Serafino Cicerone and
               Gabriele Di Stefano and
               Alfredo Navarra},
  title     = {Gathering of robots on meeting-points: feasibility and optimal resolution algorithms},
  journal   = {Distributed Comput.},
  volume    = {31},
  number    = {1},
  pages     = {1--50},
  year      = {2018}
}

@inproceedings{Viglietta13,
  author    = {G. Viglietta},
  title     = {Rendezvous of Two Robots with Visible Bits},
  booktitle = {9th International Symposium on Algorithms
               and Experiments for Sensor Systems, Wireless Networks and Distributed Robotics, (ALGOSENSORS-2013)},
  series    = {},
  volume    = {},
  pages     = {291-306},
  publisher = {},
  year      = {2013},
  url       = {},
  doi       = {},
  bibsource = {}
}

@INPROCEEDINGS{11318926,
  author={Aliberti, Andrea and Kim, Yonghwan and Katayama, Yoshiaki},
  booktitle={2025 Thirteenth International Symposium on Computing and Networking Workshops (CANDARW)}, 
  title={Limits of Fault-Tolerant Observation: Impossibility of Gathering in the (4,2)-Defected View Model}, 
  year={2025},
  volume={},
  number={},
  pages={121-127},
  doi={10.1109/CANDARW68385.2025.00029},
  ISSN={2832-1324},
  month={Nov},}

@ARTICLE{Suzuki1999,
    author = {Ichiro Suzuki  and Masafumi Yamashita.},
    title = {Distributed anonymous mobile robots: Formation of geometric patterns},
    journal = {SIAM Journal on Computing},
    year = {1999},
    volume = {28},
    pages = {pages 1347-1363}
}

@incollection{Cohen2004,
year={2004},
isbn={978-3-540-22230-9},
booktitle={Structural Information and Communication Complexity},
volume={3104},
series={Lecture Notes in Computer Science},
doi={10.1007/978-3-540-27796-5_8},
title={Robot Convergence via Center-of-Gravity Algorithms},
url={http://dx.doi.org/10.1007/978-3-540-27796-5_8},
publisher={Springer Berlin Heidelberg},
author={Cohen, R. and Peleg, D.},
pages={79-88}
}

@article{Flocchini2008,
 author = {Flocchini, P. and Prencipe, G. and Santoro, N. and Widmayer, P.},
 title = {Arbitrary pattern formation by asynchronous, anonymous, oblivious robots},
 journal = "Theoretical Computer Science",
 issue_date = {2008},
 volume = {407},
 number = {1-3},
  
 year = {2008},
 issn = {0304-3975},
 pages = {pages 412--447},
 numpages = {36},
 url = {http://dx.doi.org/10.1016/j.tcs.2008.07.026},
 doi = {10.1016/j.tcs.2008.07.026},
 acmid = {1450494},
 publisher = {Elsevier Science Publishers Ltd.},
 address = {Essex, UK},
}

@article{Prencipe2007,
title = "Impossibility of gathering by a set of autonomous mobile robots",
journal = "Theoretical Computer Science",
volume = "384",
number = "2 - 3",
pages = "pages 222--231",
year = "2007",
issn = "0304-3975",
doi = "10.1016/j.tcs.2007.04.023",
url = "http://www.sciencedirect.com/science/article/pii/S0304397507003349",
author = "Giuseppe Prencipe."
}

@INPROCEEDINGS{flocci2014,
    author = {Flocchini, P. and Prencipe, G. and Santoro, N. and Viglietta, G.},
    title = { Distributed Computing by Mobile Robots: Solving the Uniform Circle Formation Problem},
    booktitle = {The 18th International Conference on Principles of Distributed Systems (OPODIS 2014)},
    year = {2014},
    pages = {217-232},
    publisher = { }
}

@incollection{Fujinaga2010,
year={2010},
isbn={978-3-642-17652-4},
booktitle={Principles of Distributed Systems},
volume={6490},
series={Lecture Notes in Computer Science},
doi={10.1007/978-3-642-17653-1_1},
title={Pattern Formation through Optimum Matching by Oblivious CORDA Robots},
url={http://dx.doi.org/10.1007/978-3-642-17653-1_1},
publisher={Springer Berlin Heidelberg},
author={Fujinaga, N. and Ono, H. and Kijima, S. and Yamashita, M.},
pages={1-15}
}

@article{Cieliebak2012,
author = "Mark Cieliebak and Paola Flocchini and  Giuseppe Prencipe and Nicola Santoro",
title = {Distributed Computing by Mobile Robots: Gathering},
journal = {SIAM Journal on Computing},
volume = {41},
number = {4},
pages = {pages 829-879},
year = {2012},
doi = {10.1137/100796534},
URL = {http://epubs.siam.org/doi/abs/10.1137/100796534},
eprint = {http://epubs.siam.org/doi/pdf/10.1137/100796534}
}

@inproceedings{Agathangelou2012,
 author = "Chrysovalandis Agathangelou and Chryssis Georgiou and Marios Mavronicolas",
 title = {A Distributed Algorithm for Gathering Many Fat Mobile Robots in the Plane},
 booktitle = {Proc. ACM Symposium on Principles of Distributed Computing (PODC)},
 isbn = {978-1-4503-2065-8},
 location = {Montr\&\#233;al, Qu\&\#233;bec, Canada},
 pages = {250--259},
 numpages = {10},
 year = {2013},
 url = {http://doi.acm.org/10.1145/2484239.2484266},
 doi = {10.1145/2484239.2484266},
 acmid = {2484266},
}

@book{Santoro2012,
  author    = "Paola Flocchini and Giuseppe Prencipe  and Nicola Santoro",
  title     = {Distributed Computing by Oblivious Mobile Robots},
  publisher = {Morgan {\&} Claypool Publishers},
  series    = {Synthesis Lectures on Distributed Computing Theory},
  year      = {2012},
  ee        = {http://dx.doi.org/10.2200/S00440ED1V01Y201208DCT010},
  bibsource = {DBLP, http://dblp.uni-trier.de}
}

@inproceedings{FlocchiniKKSY19,
  author    = {Paola Flocchini and
               Ryan Killick and
               Evangelos Kranakis and
               Nicola Santoro and
               Masafumi Yamashita},
  editor    = {Pinyan Lu and
               Guochuan Zhang},
  title     = {Gathering and Election by Mobile Robots in a Continuous Cycle},
  booktitle = {30th International Symposium on Algorithms and Computation, {ISAAC}},
  volume    = {149},
  pages     = {8:1--8:19},
  year      = {2019},

}

@article{Izumi2013,
 author = {Izumi, T. and Izumi, T. and Kamei, S. and Ooshita, F.},
 title = {Feasibility of Polynomial-Time Randomized Gathering for Oblivious Mobile Robots},
 journal = {IEEE Trans. Parallel Distrib. Syst.},
 issue_date = {April 2013},
 volume = {24},
 number = {4},
 month = apr,
 year = {2013},
 issn = {1045-9219},
 pages = {716--723},
 numpages = {8},
 url = {http://dx.doi.org/10.1109/TPDS.2012.212},
 doi = {10.1109/TPDS.2012.212},
 acmid = {2478835},
 publisher = {IEEE Press},
 address = {Piscataway, NJ, USA},
}

@article{PattanayakMRM19,
  author    = {D. Pattanayak and
               K. Mondal and
               H. Ramesh and
               P. S. Mandal},
  title     = {Gathering of mobile robots with weak multiplicity detection in presence of crash-faults},
  journal   = {Journal of Parallel Distributed Computing},
  volume    = {123},
  pages     = {145--155},
  year      = {2019},
}

@article{czyzowicz2009gathering,
	title={Gathering few fat mobile robots in the plane},
	author={Jurek Czyzowicz and  Leszek Gasieniec and Andrzej Pelc.},
	journal={Theoretical Computer Science},
	volume={410},
	number={6},
	pages={pages 481--499},
	year={2009},
	publisher={Elsevier}
}

@article{GervasiP04,
  author    = {Vincenzo Gervasi and
               Giuseppe Prencipe},
  title     = {Coordination without communication: the case of the flocking problem},
  journal   = {Discret Applied Mathematics},
  volume    = {144},
  number    = {3},
  pages     = {324--344},
  year      = {2004}
}

@article{tomita2017plane,
  title        = {Plane Formation by Synchronous Mobile Robots without Chirality},
  author       = {Tomita, Yusuke and Yamauchi, Yoshiaki and Kijima, Shin-ichi and Yamashita, Masafumi},
  journal      = {arXiv preprint arXiv:1705.06521},
  year         = {2017}
}

@inproceedings{uehara2016plane,
  title        = {Plane Formation by Semi-Synchronous Robots in the Three-Dimensional Euclidean Space},
  author       = {Uehara, Tetsuya and Yamauchi, Yoshiaki and Kijima, Shin-ichi and Yamashita, Masafumi},
  booktitle    = {Proceedings of the International Symposium on Stabilization, Safety, and Security of Distributed Systems (SSS)},
  pages        = {383--398},
  year         = {2016}
}

@inproceedings{yamauchi2013pattern,
  title={Pattern formation by mobile robots with limited visibility},
  author={Yamauchi, Yukiko and Yamashita, Masafumi},
  booktitle={International Colloquium on Structural Information and Communication Complexity},
  pages={201--212},
  year={2013},
  organization={Springer}
}

@article{YamauchiUKY17,
  title        = {Plane Formation by Synchronous Mobile Robots in the Three-Dimensional Euclidean Space},
  author       = {Yamauchi, Yoshiaki and Uehara, Tetsuya and Kijima, Shin-ichi and Yamashita, Masafumi},
  journal      = {Journal of the ACM},
  volume       = {64},
  number       = {3},
  pages        = {16:1--16:43},
  year         = {2017}
}

@article{Flocchini2005,
title = "Gathering of asynchronous robots with limited visibility ",
journal = "Theoretical Computer Science ",
volume = "337",
number = "1-3",
pages = "147 - 168",
year = "2005",
note = "",
issn = "0304-3975",
doi = "http://dx.doi.org/10.1016/j.tcs.2005.01.001",
url = "http://www.sciencedirect.com/science/article/pii/S0304397505000149",

author = "Flocchini, P. and Prencipe, G. and Santoro, N. and Widmayer, P."

}

@article{DefagoPP20,
  author    = {Xavier D{\'{e}}fago and
               Maria Potop{-}Butucaru and
               Philippe Raipin Parv{\'{e}}dy},
  title     = {Self-stabilizing gathering of mobile robots under crash or Byzantine faults},
  journal   = {Distributed Computing},
  volume    = {33},
  number    = {5},
  pages     = {393--421},
  year      = {2020}
}

@article{kim2022gathering,
  title={Gathering Despite Defected View},
  author={Kim, Yonghwan and Shibata, Masahiro and Sudo, Yuichi and Nakamura, Junya and Katayama, Yoshiaki and Masuzawa, Toshimitsu},
  journal={arXiv preprint arXiv:2208.08159},
  year={2022}
}

@inproceedings{kim2023gathering,
  title={Gathering of mobile robots with defected views},
  author={Kim, Yonghwan and Shibata, Masahiro and Sudo, Yuichi and Nakamura, Junya and Katayama, Yoshiaki and Masuzawa, Toshimitsu},
  booktitle={26th International Conference on Principles of Distributed Systems (OPODIS 2022)},
  pages={14--1},
  year={2023},
  organization={Schloss Dagstuhl--Leibniz-Zentrum f{\"u}r Informatik}
}

@INPROCEEDINGS{8554579,
  author={Chauhuri, Sruti Gan},
  booktitle={2018 International Conference on Advances in Computing, Communications and Informatics (ICACCI)}, 
  title={Gathering of Autonomous Mobile Robots Under Limited Sensing Ranges}, 
  year={2018},
  volume={},
  number={},
  pages={909-914},
  doi={10.1109/ICACCI.2018.8554579}}

@inproceedings{chatterjee2015gathering,
  title={Gathering asynchronous swarm robots under nonuniform limited visibility},
  author={Chatterjee, Avik and Gan Chaudhuri, Sruti and Mukhopadhyaya, Krishnendu},
  booktitle={International Conference on Distributed Computing and Internet Technology},
  pages={174--180},
  year={2015},
  organization={Springer}
}

@inproceedings{braun2020local,
  title={Local gathering of mobile robots in three dimensions},
  author={Braun, Michael and Castenow, Jannik and Meyer auf der Heide, Friedhelm},
  booktitle={International Colloquium on Structural Information and Communication Complexity},
  pages={63--79},
  year={2020},
  organization={Springer}
}

@book{santorobook2,
  editor    = {Paola Flocchini and
               Giuseppe Prencipe and
               Nicola Santoro},
  title     = {Distributed Computing by Mobile Entities, Current Research in Moving and Computing},
   publisher = {Springer},
  year      = {2019},
  }

\end{document}